\documentclass[11pt]{article}
\usepackage[usenames,dvipsnames,svgnames,table]{xcolor} 

\usepackage{enumitem} 
\usepackage{rotfloat} 
\usepackage{graphicx}
\usepackage{epsfig}
\usepackage{amssymb, amsmath}
\usepackage{verbatim}
\usepackage{natbib}
\usepackage{authblk}
\usepackage{kotex}
\usepackage{multirow}
\usepackage[colorlinks]{hyperref} 
\usepackage[colorinlistoftodos, textsize=scriptsize]{todonotes} 
\usepackage{subcaption} 

\newtheorem{theorem}{Theorem}[section]

\newenvironment{proof}[1][Proof]{\begin{trivlist}
		\item[\hskip \labelsep {\bfseries #1}]}{\end{trivlist}}

\newenvironment{remark}[1][Remark]{\begin{trivlist}
		\item[\hskip \labelsep {\bfseries #1}]}{\end{trivlist}}

\newcommand{\bea}{\begin{eqnarray*}}
	\newcommand{\eea}{\end{eqnarray*}}
\newcommand{\bean}{\begin{eqnarray}}
\newcommand{\eean}{\end{eqnarray}}

\newcommand{\bfX}{{\bf X}}

\newcommand{\sg}{\Sigma}
\newcommand{\what}{\widehat}

\newcommand{\lra}{\longrightarrow}

\newcommand{\calD}{\mathcal{D}}

\newcommand{\bbP}{\mathbb{P}} 
\newcommand{\bbR}{\mathbb{R}}
\newcommand{\bbE}{\mathbb{E}}

\begin{document}
	
	\title{Bayesian joint inference for multiple directed acyclic graphs}
	\author[1]{Kyoungjae Lee}
	\affil[1]{Department of Statistics, Inha university}
	\author[2]{Xuan Cao}
	\affil[2]{Department of Mathematical Sciences, University of Cincinnati}
	
	\maketitle
	\begin{abstract}
		In many applications, data often arise from multiple groups that may share similar characteristics.
		A joint estimation method that models several groups simultaneously can be more efficient than estimating parameters in each group separately.
		We focus on unraveling the dependence structures of data based on directed acyclic graphs and propose a Bayesian joint inference method for multiple graphs.
		To encourage similar dependence structures across all groups, a Markov random field prior is adopted.
		We establish the joint selection consistency of the fractional posterior in high dimensions, and benefits of the joint inference are shown under the common support assumption.
		This is the first Bayesian method for joint estimation of multiple directed acyclic graphs.
		The performance of the proposed method is demonstrated using simulation studies, and it is shown that our joint inference outperforms other competitors. 
		We apply our method to an fMRI data for simultaneously inferring multiple brain functional networks.
	\end{abstract}
	
	Key words: Joint selection consistency, Markov random field prior, Cholesky factor
	

	\section{Introduction}

	Suppose we observe data from the following $K$ groups,
	\bean\label{model_K}
	X_{k,1} ,\ldots, X_{k, n_k}  \mid \Omega_k &\overset{ind.}{\sim}&  N_p ( 0,  \Omega_k^{-1} )  , \,\, k=1,\ldots, K ,
	\eean
	where $\Omega_k \in \bbR^{p\times p}$ is the precision matrix of the $k$th group.
	Here, $N_p(\mu, \sg)$ denotes the $p$-dimensional normal distribution with the mean vector $\mu \in \bbR^p$ and covariance matrix $\sg \in \bbR^{p\times p}$.
	We are interested in investigating the dependence structures of each multivariate data set, especially in high-dimensional settings.
	To consistently recover the dependence structure of multivariate data, various sparsity assumptions have been suggested for high-dimensional covariance matrices \citep{cai2010optimal,cai2012optimal}, precision matrices \citep{banerjee2015bayesian,ren2015asymptotic} and Cholesky factors \citep{lee2017estimating,cao2019posterior}.
	In this paper, we focus on sparse Cholesky factors, whose sparsity patterns are related to directed acyclic graph (DAG) models.
	Our goal is to develop a theoretically supported Bayesian method for jointly estimating multiple DAGs under a sparsity assumption.

	In many applications, data are collected from multiple groups that share similar characteristics.
	For examples, gene expression levels are often measured over the patients with different subtypes \citep{cai2016joint, liu2019jointDAG}, where the DAGs may vary across subtypes but share similar structures.
	Then, joint estimation can be more efficient than estimating each DAG separately.
	Another motivation for this type of problem comes from neuroimaging studies.
	In neuroimaging studies, it is common to explore the changes in functional connectivity for different brain regions through the progression of a certain disease. Taking the Parkinson's disease (PD) as an example, during the progression of PD, some patients may develop the comorbidity of depression, and others may not. Neuroscientists are interested in learning the complex interactions that govern brain connectivity networks and contribute to the onset of depression. In such applications, statistical methods for jointly estimating multiple DAGs can serve as a powerful tool to gain insight into the underlying neurological mechanism.

	
	When data are collected from a homogeneous population, many statistical methods for estimating high-dimensional sparse Cholesky factors have been developed.
	\cite{shojaie2010penalized} proposed a penalized likelihood method based on a lasso-type penalty and derived its convergence rate.
	\cite{van2013ell} showed the convergence rate of the $\ell_0$-penalized maximum likelihood estimator for sparse Cholesky factors.
	Recently, \cite{khare2019scalable} developed a convex sparse Cholesky selection, by using a reparameterization trick, and proved the convergence rate and selection consistency in a moderate high-dimensional setting.
	From a Bayesian perspective, \cite{ben2015high} introduced a class of DAG-Wishart priors for sparse DAG models, and \cite{cao2019posterior} showed the posterior convergence rate and selection consistency of hierarchical DAG-Wishart priors.
	Based on the autoregressive model representation of a Gaussian DAG model, \cite{lee2019minimax} developed an empirical sparse Cholesky prior.
	They showed that the proposed prior attains the minimax optimal posterior convergence rate as well as the selection consistency under mild conditions.
	However, the above methods are lack of sharing information across graphs when estimating multiple graphs with similar structures.

	To infer data sets from heterogeneous populations, various methods have been proposed for estimating multiple graphical models, i.e., precision matrices, by \cite{danaher2014joint}, \cite{cai2016joint}, 	\cite{peterson2015bayesian} and \cite{gan2019bayesian}, to name a few.	
	On the other hand, only few joint inference methods for multiple DAGs have been proposed in the literature.
	\cite{wang2020high} proposed the joint greedy equivalence search for estimating multiple DAGs and proved its convergence rate  under the Frobenius norm.
	They showed that the cardinality of the union of estimated DAGs has the same rate with that of the union of true DAGs. 
	Recently, \cite{liu2019jointDAG} proposed a two-step method, called the multiple PenPC, to jointly estimate the skeletons of DAGs and showed the joint selection consistency of the skeletons in high-dimensional settings.
	To the best of our knowledge, no Bayesian method, which enjoys theoretical guarantees in high-dimensional settings, has yet been suggested for multiple DAGs.

	In this paper, we propose a prior for Bayesian joint inference, called the joint empirical sparse Cholesky prior, for multiple DAGs in high-dimensional settings.
	We show that the proposed prior achieves the joint selection consistency under mild conditions, which means that the marginal posterior at the true DAGs converges to one as more data are collected (Theorem \ref{thm:selection}).
	To the best of our knowledge, this is the first work that has established the joint selection consistency for multiple DAGs under a Bayesian framework.
	We also prove theoretical benefits of the joint inference under the common support assumption. 
	Specifically, it is shown that the proposed method attains the joint selection consistency under much weaker beta-min conditions (Theorems \ref{thm:advan2} and \ref{thm:advan3}) compared with separate inferences.
	In simulation studies, our joint inference method outperforms the other state-of-the-art methods including frequentist joint estimators and Bayesian separate inferences especially in high overlapping scenarios.
	These finding support our motivation for joint inference: when multiple DAGs share similar structures, joint estimation can be more efficient than separate estimations.

	The rest of paper is organized as follows.
	Section \ref{sec:prel} introduces multiple Gaussian DAG models, the joint empirical sparse Cholesky prior and the fractional posterior distribution.
	In Section \ref{sec:main}, we show the joint selection consistency of the proposed method and benefits of the joint inference compared with separate inferences.
	The finite sample performance of our method is investigated in Section \ref{sec:simul}, and we conduct a real data analysis using a functional magnetic resonance imaging (fMRI) dataset in Section \ref{sec:real}.
	Section \ref{sec:disc} concludes the paper with a discussion.
	The proofs of the main results are given in Section \ref{sec:proof}.

	\section{Preliminaries}\label{sec:prel}

	\subsection{Multiple Gaussian DAG models}
	
	For a given precision matrix $\Omega \in \bbR^{p\times p}$, let $\Omega = (I_p - A)^T D^{-1} (I_p - A)$ be its modified Cholesky decomposition (MCD), where $A = ( a_{ jl} )$ is a lower triangular matrix with $a_{ jj} = 0$ and $D = diag(d_{j})$ with $d_{j}>0$, for all $j=1,\ldots, p$.	
	Then, it is well known that $X = (X_1,\ldots, X_p)^T \sim N_p(0, \Omega^{-1})$ can be represented as a sequence of linear autoregressive models as follows:
	\bea
	X_1 \mid d_j &\sim& N (0, d_1)  ,\\
	X_j \mid a_{S_j}, d_j, S_j &\sim& N \Big( \sum_{l \in S_j} X_l a_{jl}, d_j \Big) , \,\, j=2,\ldots, p,
	\eea
	where $a_{S_j} = (a_{j l})^T_{l \in S_j} \in \bbR^{|S_j|}$, $S_j \subseteq \{1,\ldots, j-1\}$ and $|S_j|$ is the cardinality of $S_j$ \citep{bickel2008regularized}.
	The support of the Cholesky factor, $\{S_2,\ldots, S_p\}$, determines the DAG, $\calD = (V, E)$.
	Here, $V = \{1,\ldots, p\}$ is a set of vertices, and $E$ is a set of directed edges, where $\{l \to j  \}\in E$ if and only if $a_{jl} \neq 0$.
	In this paper, we assume that a parent ordering of variables is known in which no edges exist from larger vertices to smaller vertices.
	The above model is called the Gaussian DAG model.
		
	Similarly, for a given $1\le k\le K$, we denote the MCD of $\Omega_k$ by $\Omega_k = (I_p - A_k)^T D_k^{-1} (I_p - A_k)$, where $A_k = ( a_{k, jl} )$ and $D_k = diag(d_{kj})$.
	Let $S_{kj} = ( S_{k, j1}, \ldots, S_{k, j j-1} ) \in \{0,1\}^{j-1}$ be the support of the $j$th row of $A_k$ with $S_{k, jl} = I(a_{k, jl} \neq 0 )$.
	With a slight abuse of notation, if there is no confusion, $S_{kj}$ is sometimes used to denote the set of nonzero indices in the $j$th row of $A_k$, i.e., $S_{kj} = \{ l : a_{k,jl} \neq 0 \} \subseteq \{1,\ldots, j-1\}$.
	We denote the data from the $k$th group and the whole data by $\bfX_{k} = ( X_{k,1},\ldots, X_{k, n_k})^T \in \bbR^{n_k \times p}$ and $\tilde{\bfX}_n = (\bfX_1^T, \ldots, \bfX_K^T )^T \in \bbR^{n \times p}$, respectively, where $n =\sum_{k=1}^K n_k$. 
	Then, model \eqref{model_K} can be expressed as follows:
	\bean\label{model_reg}
	\begin{split}
		\bfX_{k,1} \mid d_{kj} \,\,&\overset{ind.}{\sim}\,\, N_{n_k} ( 0, d_{kj} I_{n_k} ) , \\
		\bfX_{k,j} \mid a_{k, S_{kj}} , d_{kj} , S_{kj} \,\,&\overset{ind.}{\sim}\,\,  N_{n_k} \Big( \bfX_{k,S_{kj}} a_{k, S_{kj}} , \,\, d_{kj} I_{n_k}  \Big) , \,\, j=2,\ldots, p ,\,\, k=1,\ldots, K,
	\end{split}
	\eean
	where $a_{k, S_{kj}} = (a_{k, jl} )^T_{l \in S_{kj}} \in \bbR^{|S_{kj}|}$ and $\bfX_{k, S} \in \bbR^{n_k \times |S|}$ is the submatrix consisting of $S$th columns of $\bfX_{k}$ for any $S \subseteq \{1,\ldots, j-1\}$.
	We call model \eqref{model_reg} the multiple Gaussian DAG models.
	Note that the lower triangular part of  $A_k$ can be seen as a set of regression vectors, thus we can use a prior tailored to each row of the sparse regression coefficient vectors.
	We assume that the sample size for each group, $n_k$, can be different across all groups.
	We consider the high-dimensional setting in which $p \ge n$ and allow the number of groups, $K$, grow to infinity as we observe more data.

	\subsection{Joint empirical sparse Cholesky priors}

	\cite{lee2019minimax} proposed the empirical sparse Cholesky (ESC) prior for a sparse DAG model on the basis of the interpretation \eqref{model_reg}.
	In this paper, we extend this prior to deal with multiple DAGs.	
	For given $1\le k \le K$ and $S_{kj}$, we use the following conditional prior for $A_k$ and $D_k$:
	\bean\label{a_and_d}
	\begin{split}
		a_{k , S_{kj}} \mid d_{kj} , S_{kj}  \,\,&\overset{ind.}{\sim}\,\,  N_{|S_{kj}|} \Big(  \what{a}_{k, S_{kj}}  , \,\, \frac{d_{kj} }{\gamma} \big( \bfX_{k, S_{kj}}^T \bfX_{k, S_{kj}}  \big)^{-1}  \Big)   ,  \quad j=2,\ldots, p, \\
		\pi(d_{kj}) \,\,&\propto\,\,   d_{kj}^{ - \nu_0/2 -1 } , \quad  j=1,\ldots, p,
	\end{split}
	\eean
	for some positive constants $\gamma$ and $\nu_0$, where $\what{a}_{k, S_{kj}} = ( \bfX_{k, S_{kj}}^T \bfX_{k, S_{kj}} )^{-1} \bfX_{k, S_{kj}}^T \bfX_{k,j} $.
	This corresponds to the ESC prior when $K=1$.	
	Note that the conditional prior for $a_{k, S_{kj}}$ is an empirical version of the Zellner's $g$-prior \citep{zellner1986assessing} centered at $\what{a}_{k, S_{kj}}$, and the prior for $d_{kj}$ becomes the Jeffreys prior \citep{jeffreys1946invariant} when $\nu_0 =0$.
	
	For joint inference on multiple DAGs, given an integer $j \in \{2,\ldots, p\}$, we propose the following joint prior for $(S_{1j} ,\ldots, S_{Kj})$:
	\bean
	\pi( S_{1j},\ldots, S_{Kj} ) &\propto&  f( S_{1j}, \ldots, S_{Kj} )  \prod_{k=1}^K \pi(S_{kj})  , \label{joint_S_prior}
	\eean
	where 
	\bea
	\pi(S_{kj}) &\propto&  \binom{j-1}{|S_{kj}|}^{-1} p^{- c_1 |S_{kj}|} I( 0\le |S_{kj}| \le R_j)  
	\eea
	for some positive integers $0< R_j \le j-1$.
	Here, $\pi(S_{kj}) $ plays a role as a penalty term for the model size $|S_{kj}|$, which prefers sparse models.
	Similar priors have been used in the literature including \cite{martin2017empirical} and \cite{lee2019minimax}.
	For $f( S_{1j}, \ldots, S_{Kj} )$ in \eqref{joint_S_prior}, we suggest using the following Markov random field (MRF) type prior to reflect the expectation that different groups share similar DAG structures:
	\bea
	f( S_{1j}, \ldots, S_{Kj} )  &=& \exp \Big\{   c_{2j} \sum_{l=1}^{j-1} \tilde{S}_{jl}^T (1_K 1_K^T - I_K ) \tilde{S}_{jl}   \Big\} \\
	&=&   \exp \Big\{  2 c_{2j} \sum_{l=1}^{j-1} \sum_{k < k'} I( S_{k, jl} = S_{k', jl} = 1 )     \Big\}    ,\quad  j=2,\ldots, p 
	\eea
	for some constant $c_{2j}>0$, where $\tilde{S}_{jl} = (S_{1, jl} ,\ldots, S_{K, jl} )^T$ and $1_K = (1,\ldots, 1)^T \in \bbR^K$.
	This MRF prior encourages similar patterns of sparsity for $(S_{1j}, \ldots, S_{Kj})$.
	\cite{peterson2015bayesian} used a similar MRF prior for inferring multiple graphical models.
	By putting together priors \eqref{a_and_d} and \eqref{joint_S_prior}, we propose a prior for multiple DAGs,
	\bea
	\pi( \Omega_1, \ldots, \Omega_K  ) 
	&\propto&  \prod_{j=2}^p \pi( S_{1j}, \ldots, S_{Kj} )    
	\prod_{k=1}^K \Big\{  \prod_{j=2}^p  \pi(a_{k, S_{kj}} \mid d_{kj}, S_{kj}) \prod_{j=1}^p  \pi(d_{kj}) \Big\}  ,
	\eea 
	which we call the joint empirical sparse Cholesky (JESC) prior hearafter.
	
	
	\subsection{$\alpha$-fractional posterior}
	
	We adopt the fractional likelihood framework, which has received increasing attention in recent years \citep{martin2014asymptotically,martin2017empirical,lee2019minimax}.
	Let $\theta$ and $L(\theta)$ be a parameter and a likelihood function, respectively.
	For a given constant $\alpha \in (0,1)$, $\alpha$-fractional likelihood $L_\alpha(\theta)$ is the likelihood with power $\alpha$, i.e., $\{L(\theta)\}^\alpha$.	
	Based on the JESC prior and $\alpha$-fractional likelihood, we have the following posterior distributions:
	\bea
	a_{k, S_{kj}} \mid d_{kj} , S_{kj}, \bfX_{k} &\overset{ind.}{\sim}& N_{|S_{kj}|} \Big( \what{a}_{k, S_{kj}} ,\,\,  \frac{d_{kj}}{\alpha + \gamma} \big( \bfX_{k, S_{kj}}^T \bfX_{k, S_{kj}} \big)^{-1} \Big)  ,\,\, j=2,\ldots, p , \\
	d_{kj} \mid S_{kj} , \bfX_k &\overset{ind.}{\sim}&  IG \Big(  \frac{\alpha n_k + \nu_0}{2} , \,\, \frac{\alpha n_k}{2} \what{d}_{k, S_{kj}} \Big)  , \,\, j=1,\ldots, p,
	\eea
	and
	\bea
	\pi_\alpha ( S_{1j},\ldots, S_{Kj}  \mid \tilde{\bfX}_n )  &\propto& \pi(S_{1j},\ldots, S_{Kj}) \prod_{k=1}^K f_\alpha ( \bfX_{n_k} \mid S_{kj}) , \,\, j=2,\ldots, p, 
	\eea
	where $\widehat d_{k, S_{kj}} = n_k^{-1} \bfX_{k,j}^T(I_{n_k} - \tilde P_{S_{kj}}) \bfX_{k,j}$, $\tilde P_{S_{kj}} = \bfX_{k,S_{kj}}( \bfX_{k,S_{kj}}^T \bfX_{k,S_{kj}})^{-1} \bfX_{k,S_{kj}}^T$ and 
	\bea
	f_\alpha( \bfX_{n_k} \mid S_{kj} )  &=&   \iint   L_\alpha (a_{k, S_{kj}} , d_{kj}, S_{kj}) \pi( a_{k, S_{kj}} \mid d_{kj} , S_{kj}) \pi( d_{kj} \mid S_{kj}) d a_{k, S_{kj}} \, d d_{kj} \\
	&\propto& \Big(1 + \frac{\alpha}{\gamma} \Big)^{- \frac{|S_{kj}|}{2}}  (\what{d}_{k, S_{kj}})^{-\frac{\alpha n_k + \nu_0}{2} } .
	\eea
	We denote the posterior by $\pi_\alpha (\cdot \mid \tilde{\bfX}_n)$ to indicate that the $\alpha$-fractional likelihood is used, and call it the $\alpha$-fractional posterior.
	To conduct the posterior inference for $(S_{1j},\ldots, S_{Kj})$, the Metropolis-Hastings within Gibbs algorithm can be used.
	The details are given in Section \ref{subsec:posterior}.
	Once we have posterior samples of $(S_{1j},\ldots, S_{Kj})$, the posterior samples of $a_{k, S_{kj}} $ and  $d_{kj}$ can be directly drawn from the normal and inverse-gamma distributions, respectively.

	\section{Main Results}\label{sec:main}
	
	\subsection{Joint selection consistency}\label{subsec:joint}
	
	In this section, we establish the joint selection consistency of the proposed JESC prior, which guarantees that we can recover the true DAGs asymptotically.
	Let $\Omega_{0k}$ be the true precision matrix of the $k$th class, for $k=1,\ldots, K$.	
	Let $\Omega_{0k} = (I_p - A_{0k})^T D_{0k}^{-1}(I_p - A_{0k})$ be the MCD of $\Omega_{0k}$, where $A_{0k} = ( a_{0k ,jl} )$ and $D_{0k} = diag( d_{0k, j} )$.
	We denote $S_{A}$ as the support of the matrix $A= (a_{jl})$, i.e., $S_{A} = ( I(a_{jl} \neq 0) )$.
	We first introduce the following sufficient conditions for true parameters:	\\
		
	\noindent{\bf{Condition (A1)}} There exists a constant $0<\epsilon_0 <0.5$ such that $\epsilon_0 \le \min_{1\le k\le K} \lambda_{\min}(\Omega_{0k}) \le \max_{1\le k\le K} \lambda_{\max}(\Omega_{0k}) \le \epsilon_0^{-1}$.
	
	\noindent{\bf{Condition (A2)}} $\max_{1\le k\le K} \max_{2\le j \le p} \sum_{l=1}^{p} I( a_{0k, jl} \neq 0 )  \le s_0$ for some $1\le s_0 \le p$.
	
	\noindent{\bf{Condition (A3)}} For some constant $C_{\rm bm}>0$,
	\bea
	\min_{1\le k \le K} \min_{(j,l): a_{0k,jl} \neq 0 } n_k \, a_{0k, jl}^2   &\ge&  \frac{16}{\alpha(1-\alpha) \epsilon_0^2 (1-2 \epsilon_0)^2 } C_{\rm bm} \log p.
	\eea
	
	\noindent{\bf{Condition (A4)}} $K = o(\log p)$.\\

	Condition (A1) implies that the eigenvalues of each precision matrix $\Omega_{0k}$ are bounded.
	This condition is used to obtain upper bounds of $d_{0k,j}, d_{0k, j}^{-1}$ and $\|A_{0k}\|$.
	Similar conditions have been used in, for examples, \cite{ren2015asymptotic}, \cite{khare2019scalable} and \cite{lee2019minimax}.

	Condition (A2) controls the maximum number of nonzero entries in each row of $A_{0k}$.
	This condition allows the upper bound $s_0$ to grow to infinity as $n$ get larger.
	Note that the estimation of each row of $A_{0k}$ can be considered as the estimation of regression coefficient vector, thus introducing this condition seems natural.

	Condition (A3) is the well known {\it beta-min} condition for the minimum nonzero entries of each Cholesky factor, $A_{0k}$.
	This roughly means that the lower bound for nonzero $a_{0k, jl}^2$ is of order $O( \log p /n_k )$.
	The beta-min condition is essential for consistent variable selection in high-dimensional linear regression models \citep{martin2017empirical, yang2016computational} and Gaussian DAG models \citep{yu2017learning,cao2019posterior}.
	Note that if we assume $k=1$ and $n_k = n$, then the rate of the lower bound in condition (A3) becomes $\log p/n$, which is the best (minimum) beta-min condition in the literature.

	Condition (A4) restricts the number of classes.
	Note that $K$ can grow to infinity as $n\to\infty$ at a rate slower than $\log p$.
	\cite{cai2016joint} and \cite{wang2020high} used similar condition for joint estimation of high-dimensional precision matrices and DAGs, respectively.  \\

	\noindent{\bf{Condition (P)}} $\nu_0 = o( \min_k n_k ), c_1 >2, c_{2j} \le 1/(j-1)$ and $\gamma = O(1)$. 
	For some small $0 < c_3  < (\epsilon')^2 \epsilon_0^2 / \{ 128(1+2 \epsilon_0)^2 \}$ and $\epsilon' = \{(1-\alpha)/10 \}^2$, we assume that $R_j = \lfloor \{ (\log n)^{-1} \vee c_3 \} \min_k n_k /\log p \rfloor $. \\
	
	Condition (P) shows a sufficient condition for hyperparameters in the JESC prior to obtain the desired theoretical properties, where ``P'' stands for ``prior''.
	The constant $c_1$ controls the penalty for the sparsity of Cholesky factors, thus the condition $c_1>2$ gives the minimum strength of the penalty.
	The constant $c_{2j}$ in the MRF prior controls the penalty for similarities across the DAGs, thus $c_{2j} \le 1/(j-1)$ implies that the effect of the MRF prior should not be too strong.
	This intuitively makes sense because if $c_{2j}$ is too large and dominates the other priors and likelihoods, then the posterior will always select the full model, i.e., $S_{kj} = \{ 1,\ldots, j-1 \}$ for all $1\le k \le K$ and $2\le j \le p$.
	The condition $R_j = \lfloor \{ (\log n)^{-1} \vee c_3 \} \min_k n_k /\log p \rfloor $ implies that the maximum number of nonzero entries in each row of $A_{0k}$ should at least be of order $\min_k n_k / \log p$ for the consistent selection.
	In finite samples, we suggest choosing $R_j = \lfloor \min_k n_k (\log p \, \log n)^{-1} \rfloor$.
	In Section \ref{subsec:posterior}, we will give a practical guidance for the choice of hyperparameters.

	\begin{theorem}[Joint selection consistency]\label{thm:selection}
		Suppose that conditions (A1)-(A4) and (P) hold with $C_{\rm bm} > c_1 +2$.
		Then, if $s_0 \log p \le \min_k n_k c_3/2$ and $s_0 \ge C_{\rm bm} - c_1 -1$, we have
		\bea
		\bbE_0 \Big\{  \pi_\alpha \Big( S_{A_1} = S_{A_{01} }, \ldots, S_{A_K} = S_{A_{0K}}  \mid \tilde{\bfX}_n    \Big)   \Big\}   &\lra& 1 \quad \text{ as }  \min_k n_k\to\infty.
		\eea
	\end{theorem}

	Theorem \ref{thm:selection} presents the joint selection consistency for multiple DAGs.
	It is worth comparing our result with those in \cite{liu2019jointDAG} in terms of the required conditions.
	To obtain consistency, they assumed $\min_k \lambda_{\min}( \sg_{0k, AA} ) \ge C_1$ and $\max_k \sg_{0k,jj} < C_2$ for any $A \in \{1,\ldots, p\}$ with $|A|\le q$ and some constants $C_1$ and $C_2>0$, which is weaker than condition (A1), where $\sg_{0k, AA} = ( \sg_{0k, jl} )_{j,l \in A}$.
	It was also assumed that $n_k \asymp n$ for all $k=1,\ldots,K$. 
	Note that this condition implies  $K= O(1)$, thus it is stronger than our condition (A4).
	They further assumed $p = O( \exp (n^a) )$ and $q = \max_j |A_j | = O(n^b)$ for some constants $a \in[0,1)$ and $b \in [0, (1-a)/2)$, where $A_j = \cup_{k=1}^K A_j^{(k)}$ and $A_j^{(k)} = \{l  : \Omega_{0k, jl} \neq 0 \text{ and } l \neq j \}$.
	By Lemma 1 in \cite{liu2019jointDAG}, $q = O(n^b)$ implies  $s_0 \le |\cup_{k=1}^K S_{0k,j}| = O(n^b)$, thus it is slightly more restrictive than our conditions, (A2) and $s_0 \log p \le \min_k n_k c_3/2$.
	They used the beta-min conditions,
	\bean\label{betamin_liu}
	\min_{1\le k \le K} \min_{(j,l): \Omega_{0k,jl} \neq 0 }  \Big| \frac{\Omega_{0k,jl}}{\Omega_{0k, jj}}  \Big|  &\gtrsim& n^{-d_1} \quad \text{ and } \quad n^{-d_2} \lesssim | \rho_{jl | S}^{(k)} | \le M < 1 ,
	\eean
	for some $0<d_1 < (1-a-b)/2$, $0< d_2 < \{1- (a\vee b)\}/2$ and  any $S \in \Pi_{jl}^{(k)}$, where $\Pi_{jl}^{(k)} = \{ A_{jl}^{(k)} \setminus D_{jl}^{(k)} : D_{jl}^{(k)} \subseteq C_{jl}^{(k)}  \}$, $A_{jl}^{(k)}$ is the Markov blanket of $j$ and $l$ after removing their common children and descendants, and $C_{jl}^{(k)}$ is the set of common children or descendants.
	Note that \eqref{betamin_liu} consists of two beta-min conditions to guarantee selection consistency in each step.
	Although their beta-min conditions are not directly comparable with ours, the squares of the lower bounds in \eqref{betamin_liu}, $n^{-2d_1}$ and $n^{-2d_2}$, are much larger than $\log p/n_k $ in condition (A3).
	Therefore, we obtain the joint selection consistency under weaker conditions on $K$, $s_0$ and minimum nonzero signals than those in \cite{liu2019jointDAG}.

	\begin{theorem}\label{thm:advan1}
		Let $\pi^I( S_{kj } \mid \bfX_{n_k}) \propto f_\alpha (\bfX_{n_k} \mid S_{kj}) \pi(S_{kj}) $ be the independence posterior for $S_{kj}$.
		Suppose that there exists $1\le k \le K$ such that $\cup_{k' \neq k} S_{0k',j} \subseteq S_{0k,j} $.
		Then, for any $j=2\ldots, p$, we have 
		\bean\label{cond_prob_inc}
		\pi_\alpha \big(  S_{0k, j} \mid S_{01,j},\ldots, S_{0 k-1, j} , S_{0 k+1, j},\ldots, S_{0K, j} , \tilde{\bfX}_n  \big)   
		&\ge& \pi_\alpha^I ( S_{0k, j}  \mid \bfX_{n_k} )  .
		\eean
	\end{theorem}
	
	Theorem \ref{thm:advan1} shows that the joint inference increases the conditional posterior probability at the true DAGs compared to the separate inference.
	Note that $\cup_{k' \neq k} S_{0k',j} \subseteq S_{0k,j} $ holds if and only if 
	\bean\label{f_inc}
	f( S_{01,j},\ldots, S_{0 k-1, j}, S_{kj} , S_{0 k+1, j},\ldots, S_{0K, j} ) &\le& f( S_{01,j},\ldots, S_{0 k-1, j} , S_{0k,j }, S_{0 k+1, j},\ldots, S_{0K, j} )
	\eean
	for any $S_{kj} \neq S_{0k,j}$.
	For example, \eqref{f_inc} trivially holds if we assume the common support, i.e., $S_{01,j} = \cdots = S_{0K,j}$.

	\subsection{Benefits of joint inference}\label{subsec:benefit}
	
	In this section, the theoretical benefits of the joint inference, compared with  separate inferences, are presented.
	Although investigating benefits of the joint inference under heterogeneous DAGs is important, it is very challenging to explore every possible scenario.
	Thus, we focus on the case where all Cholesky factors share a common support, i.e., all DAGs share a common structure.
	For example, \cite{cai2016joint} and \cite{gan2019bayesian} also used the common support assumption for multiple precision matrices and showed advantages of the joint estimation.
	In this case, we suggest using the restricted posterior to the space of common supports,
	\bea
	\tilde{\pi}_\alpha ( S_{A} \mid \tilde{\bfX}_n )  &=& \frac{\pi_\alpha ( S_{A_1}= \cdots =S_{A_K} = S_{A} \mid \tilde{\bfX}_n ) }{\sum_{S_{A}} \pi_\alpha ( S_{A_1}= \cdots =S_{A_K} = S_{A} \mid \tilde{\bfX}_n )  }  .
	\eea
	To prove the joint selection consistency of $\tilde{\pi}_\alpha ( S_{A} \mid \tilde{\bfX}_n )$, we introduce a weakened beta-min condition as follows: \\	
	
	\noindent{\bf{Condition (B3)}} For some constant $C_{\rm bm}>0$,
	\bea
	 \min_{(j,l): a_{01,jl} \neq 0 }  \sum_{k=1}^K n_k \, a_{0k,jl}^2   &\ge&  \frac{16}{\alpha(1-\alpha) \epsilon_0^2 (1-2 \epsilon_0)^2 } C_{\rm bm}  K \, \log p  .
	\eea
	
	Note that condition (A3) implies (B3), thus we call condition (B3) a weakened beta-min condition.
	If we assume that $n_1=\cdots = n_K$, then condition (B3) roughly means that the lower bound for $K^{-1}\sum_{k=1}^K \min_{(j,l): a_{0k,jl} \neq 0 } a_{0k,jl}^2$ is of order $O(\log p / n_k)$.
	Thus, we can consistently recover the true support as long as the {\it average} of minimum signals is significant, even if  minimum signals of some classes are quite small.
	This can be seen as the benefit of the joint inference, and the following theorem states the desired result.

	\begin{theorem}[Benefit of joint inference]\label{thm:advan2}
		Assume that $S_{A_{01}} = \cdots = S_{A_{0K}} \equiv S_{0}$ and $ K \log p = o( \min_k n_k )$.
		Then, under the same condition with Theorem \ref{thm:selection}, except using condition (B3) instead of (A3), we have
		\bea
		\bbE_0 \big\{ \tilde{\pi}_\alpha ( S_{A}  = S_0 \mid \tilde{\bfX}_n )   \big\} &\lra& 1 \quad \text{ as } \min_k n_k\to\infty .
		\eea 
	\end{theorem}
	
	Note that $ K \log p = o( \min_k n_k )$ trivially holds if we assume $(\log p)^2 = o( \min_k n_k )$, by condition (A4).
	\cite{cai2016joint} assumed $K^{2a - 1} \log p \, (\log n)^2 = o( \min_k n_k)$ and $\max(K , K^{4-a}\log K ) = o(\log p)$ for some constant $a>0$.
	The second condition is comparable to our condition (A4) when $a =3$, and then the first condition becomes $K^5 \log p \, (\log n)^2 = o(\min_k n_k)$.
	Thus, our condition $ K \log p = o( \min_k n_k )$ is much weaker than that of \cite{cai2016joint}.

	In fact, if we slightly modify the prior for $S_A$, we can further weaken the beta-min condition.
	Define the modified prior for $(S_{1j},\ldots, S_{Kj})$ as
	\bea
	\tilde{\pi}(S_{1j},\ldots, S_{Kj}) &\propto& \pi(S_{1j},\ldots, S_{Kj})^{1/K}  \\
	&\propto& f(S_{1j},\ldots, S_{Kj})^{1/K} \prod_{k=1}^K \pi(S_{kj})^{1/K}  \\
	&\equiv& \tilde{f}(S_{1j},\ldots, S_{Kj}) \prod_{k=1}^K \tilde{\pi}(S_{kj}) ,
	\eea
	and let $\tilde{\pi}_\alpha^* ( S_{A} \mid \tilde{\bfX}_n )  = \tilde{\pi}_\alpha^* (S_{A_1}=\cdots = S_{A_K} = S_{A} \mid \tilde{\bfX}_n ) $ be the restricted posterior to the space of common supports using the prior $\tilde{\pi}(S_{1j},\ldots, S_{Kj})$ instead of $\pi(S_{1j},\ldots, S_{Kj})$.
	Then, it suffices to assume the following condition (C3) instead of condition (B3) to obtain the joint selection consistency: \\
	
	\noindent{\bf{Condition (C3)}} For some constant $C_{\rm bm}>0$,
	\bea
	\min_{(j,l): a_{01,jl} \neq 0 }  \sum_{k=1}^K  n_k \, a_{0k,jl}^2   &\ge&  \frac{16}{\alpha(1-\alpha) \epsilon_0^2 (1-2 \epsilon_0)^2 } C_{\rm bm}  \log p  .
	\eea
	
	\begin{theorem}[Benefit of joint inference II]\label{thm:advan3}
		Assume that $S_{A_{01}} = \cdots = S_{A_{0K}} \equiv S_{0}$.
		Then, under the same condition with Theorem \ref{thm:selection}, except using condition (C3) instead of (A3), we have
		\bea
		\bbE_0 \big\{ \tilde{\pi}_\alpha^* ( S_{A}  = S_0 \mid \tilde{\bfX}_n )   \big\} &\lra& 1 \quad \text{ as } \min_k n_k\to\infty .
		\eea 
	\end{theorem}
	
	Theorem \ref{thm:advan3} shows the advantage of the joint inference based on the restricted posterior $\tilde{\pi}_\alpha^* ( S_{A} \mid \tilde{\bfX}_n )  $: it only requires condition (C3), which is much weaker than condition (B3). 
	Compared with Theorems \ref{thm:selection} and \ref{thm:advan2}, it reveals that, under the common support assumption, we can obtain the joint selection consistency as long as the {\it summation} of minimun signals is significant.
	Note that the lower bound in condition (C3) coincides with that in (A3).
	\cite{cai2016joint} used a similar beta-min condition to condition (C3) for the nonzero entries of precision matrices, but using $\log K \, \log p$ instead of $\log p$.
	Hence, our beta-min condition is weaker than their in terms of the rate.
	Also note that $\prod_{k=1}^K \tilde{\pi}(S_{kj}) \propto \pi(S_{1j}) I(S_{1j}=\cdots= S_{Kj} )$ when $S_{1j}=\cdots = S_{Kj}$.
	Thus, this implies that it is sufficient to use a single penalty (prior) for all $K$ classes rather than use a penalty for each class.

\section{Simulation Studies} \label{sec:simul}
In this section, we carry out simulation studies to illustrate the model selection performance of our method and show its potential benefits over other contenders.
\subsection{Posterior inference} \label{subsec:posterior}
The use of the JESC prior not only guarantees the asymptotic properties but also allows us to easily conduct the posterior inference. Recall that for $j = 2, \ldots, p$,
\begin{eqnarray*}
	&&  \pi_\alpha(S_{1j}, \ldots, S_{Kj} \mid \tilde{\bfX}_n)  \\
	&\propto& \prod_{k = 1}^{K} \Big(1 + \frac \alpha \gamma\Big)^{-\frac{|S_{kj}|} 2} \big(\widehat d_{k, S_{kj}}\big)^{-\frac{\alpha n_k + \nu_0}2} \binom{j - 1}{|S_{kj}|}^{-1}   p^{-c_1 |S_{kj}|}  I\big(0 \le |S_{kj}| \le R_j\big) \\
	&& \times \exp\bigg\{c_2\sum_{l = 1}^{j - 1}\tilde S_{jl}^T(1_K1_K^T - I_K)\tilde S_{jl}\bigg\},
\end{eqnarray*}
where $\tilde S_{jl} = (S_{1,jl}, \ldots, S_{K,jl})^T$. Hence, we can run the Metropolis-Hastings within Gibbs sampling algorithm for each $j = 2, \ldots, p$ in parallel. Here, we briefly summarize the algorithm used for the inference:
\begin{enumerate}
	{\setlength\itemindent{-15pt} \item[] Run the following steps for $j = 2, \ldots, p$.} 
	\item  Set the initial values $S_{1j}^{(1)}, \ldots, S_{Kj}^{(1)}$.
	\item For each $t = 2, \ldots, T$, run the following steps for $k = 1, \ldots, K$.
	\begin{enumerate}
		\item sample $S_{kj}^{new} \sim q\big(\cdot\mid S_{kj}^{(t)}\big)$;
		\item set $S_{kj}^{(t)} = S_{kj}^{new}$ with the probability
		\bea
		&& \min \Bigg\{1, \frac{\pi_\alpha(S_{kj}^{new}\mid S_{1j}^{(t)}, \ldots, S_{k-1,j}^{(t)}, S_{k+1,j}^{(t-1)}, \ldots, S_{Kj}^{(t-1)}, \tilde{\bfX}_n)q(S_{kj}^{(j-1)}\mid S_{kj}^{new})}{\pi_\alpha(S_{kj}^{(t-1)}\mid S_{1j}^{(t)}, \ldots, S_{k-1,j}^{(t)}, S_{k+1,j}^{(t-1)}, \ldots, S_{Kj}^{(t-1)}, \tilde{\bfX}_n )q(S_{kj}^{new}\mid S_{kj}^{(t-1)})}\Bigg\} \\
		&=& \min \Bigg\{1, \frac{\pi_\alpha(S_{kj}^{new}\mid  {\bfX}_{n_k}) f(S_{1j}^{(t)}, \ldots, S_{k-1,j}^{(t)}, S_{kj}^{new}, S_{k+1,j}^{(t-1)}, \ldots, S_{Kj}^{(t-1)}) q(S_{kj}^{(j-1)}\mid S_{kj}^{new})}{\pi_\alpha(S_{kj}^{(t-1)}\mid {\bfX}_{n_k} ) f(S_{1j}^{(t)}, \ldots, S_{k-1,j}^{(t)}, S_{kj}^{(t-1)},  S_{k+1,j}^{(t-1)}, \ldots, S_{Kj}^{(t-1)} ) q(S_{kj}^{new}\mid S_{kj}^{(t-1)})}\Bigg\} ,
		\eea
		otherwise set $S_{kj}^{(t)} = S_{kj}^{(t-1)}$.
	\end{enumerate}
\end{enumerate}
The kernel $q(S^{new} \mid S)$ is chosen to form a new set $S^{new}$ by changing a randomly selected nonzero component to 0 with probability 0.5 or by changing a randomly selected zero component to 1 with probability 0.5. Steps 1 and 2 in the above algorithm, can be parallelized for each column. For more details, we refer the interested readers to \cite{cao2019posterior} and \cite{lee2019minimax}.

The tuning parameters are chosen as suggested in \cite{martin2017empirical} and \cite{lee2019minimax}. 
Specifically, we set $\alpha = 0.999$ to mimic the Bayesian model with the original likelihood. 
In practice, as long as $1 - \alpha$ is close to zero, the performance was not sensitive to the choice of $\alpha$. 
The other hyperparameters were chosen as $\gamma = 0.1$, $\nu_0 = 0$, $c_1 = 2$ and $c_{2j} = \{p(K - 1)\}^{-1}$ for $j = 2, \ldots, p$ to satisfy the theoretical conditions.
The above algorithm is coded in \verb|R| and publicly available at \verb|https://github.com/xuan-cao/Multiple-DAG-Selection|.

\subsection{Simulation setting} \label{subsec:simsetting}
In this section, we demonstrate the performance of the proposed method in various settings similar to those used in \cite{liu2019jointDAG,peterson2015bayesian,peterson2020joint}. We construct three Cholesky factors $A_1$, $A_2$, and $A_3$ corresponding to DAGs $\mathcal D_1$, $\mathcal D_2$ and $\mathcal D_3$ with different degrees of shared structure. We include $p = 150$ nodes, and consider the first scenario as follows. For the first $p \times p$ lower triangular matrix $A_1$, we randomly chose 2\% of the lower triangular entries of $A_1$ and sampled their values from a uniform distribution on $[-0.7, -0.3] \cup [0.3, 0.7]$. The remaining entries were set to zero. $\mathcal D_1$ can be acquired by mapping the nonzero entries in $A_1$ to a DAG with $p$ nodes. To obtain $\mathcal D_2$, five edges are removed from $\mathcal D_1$ and five new edges added at random. To obtain $\mathcal D_3$, five edges are removed from the graph for group 2, and five edges added at random. All the lower triangular entries in $A_2$ and $A_3$ are generated in a similar manner as in $A_1$. We call this simulation setting Scenario 1 ({\it high overlapping}), where each pair of DAGs have 218 of 223 edges (97.76\%) in common.

Next, we investigate a different simulation scenario, say Scenario 2 ({\it medium overlapping}), where $A_1, \mathcal D_1, A_2, \mathcal D_2$ are formed as in Scenario 1, but we change the design of $A_3$ and $\mathcal D_3$ as follows. To obtain $\mathcal D_3$, 20 edges are removed from the graph for group 2, and 20 edges added at random. All the entries in three Cholesky factors $A_1$, $A_2$, and $A_3$ are generated as in Scenario 1. Under this setting, $\mathcal D_1$ and $\mathcal D_2$ share 218 of 223 edges (97.76\%), $\mathcal D_2$ and $\mathcal D_3$ share 203 edges (91.03\%), and $\mathcal D_1$ and $\mathcal D_3$ share around 219 edges (89.24\%). For our final simulation setting, Scenario 3 ({\it low overlapping}), we first create $A_1$ and $\mathcal D_1$ as previously mentioned, and obtain $\mathcal D_2$ by randomly removing 20 edges and adding 20 edges from $\mathcal D_1$. $\mathcal D_3$ is again acquired by randomly removing 20 edges and adding 20 edges from $\mathcal D_2$. These steps result in DAGs $\mathcal D_1$ and $\mathcal D_2$ that share 203 of 223 edges (91.03\%), $\mathcal D_2$ and $\mathcal D_3$ that share 203 edges (91.03\%), and $\mathcal D_1$ and $\mathcal D_3$ that have 185 common edges (82.96\%). All the nonzero entries in $A_1$, $A_2$, and $A_3$ are then sampled from a uniform distribution as elaborated in Scenario 1. For all settings, we simulate the diagonal entries of $D_1, D_2, D_3$ from a uniform distribution on $[2,5]$. Given the precision matrices $\Omega_k = (I_p - A_k)^T D_k^{-1} (I_p - A_k)$ for $k = 1, 2, 3$, the data sets were generated from the multivariate normal distribution $N_p(0,\Omega_k^{-1})$ with $(n_k, p) = (100, 150)$ for $k = 1, 2, 3$.

\subsection{Performance comparison}  \label{subsec:simcomparison}
We compare the following methods: the proposed JESC prior, Bayesian inference based on ESC applied separately for each group (SESC) \citep{lee2019minimax}, multiple PenPC (MPenPC) \citep{liu2019jointDAG}, joint graphical lasso (JGL) \citep{danaher2014joint}, and seperate DAG lasso (DAGL) for each group \citep{shojaie2010penalized}. The tuning parameters in JGL were selected using a grid search to identify the combination that minimizes the AIC as suggested in \cite{danaher2014joint}. Since for our simulation studies, JGL could not produce exact zeros in the Cholesky factors of the estimated precision matrices, we further adopt the hard thresholding of these Cholesky factors. The penalty parameters in MPenPC were tuned using the  extended BIC (EBIC) \citep{chen2008} as suggested in \cite{liu2019jointDAG}. The penalty parameters in DAGL were set as $\lambda_i(\alpha) = 2 n^{- 1/2}Z^*_{{0.1}/\{2p(i-1)\}}$ (separate for each variable $i$), where $Z_q^*$ denotes the $(1-q)^{th}$ quantile of the standard normal distribution. This choice is justified in \cite{shojaie2010penalized} based on asymptotic considerations. For Bayesian methods, we ran the Metropolis-Hastings algorithm specified in Section \ref{subsec:posterior} for each data set to conduct posterior inferences. Every MCMC chain started from an empty initial state and ran for 5,000 iterations with a burn-in period of 1,000, since we observed that on average the posterior samples converged rapidly and stabilized after 1,000 iterations. The hyperparameter $c_2$ was set to 0 when implementing SESC. We constructed the final model by collecting indices with inclusion probabilities exceeding 0.5.

To evaluate the performance of joint DAG selection, the true positive rate (TPR), false positive rate (FPR), Matthews correlation coefficient (MCC), and area under the curve (AUC) are reported at Tables \ref{table:comp1}, \ref{table:comp2} and \ref{table:comp3} averaged over 20 repetitions.
The criteria are defined as
\bea
\text{TPR}  &=&     \frac{TP}{TP+FN} ,   \\
\text{FPR}  &=&    \frac{FP}{TN+FP}   ,  \\
\text{MCC}  &=&     \frac{TP \times TN - FP\times FN}{\sqrt{(TP+FP)(TP+FN)(TN+FP)(TN+FN)}}  ,	 
\eea
where TP, TN, FP and FN are true positive, true negative, false positive and false negative, respectively.
The AUC is calculated based on the TPR and the FPR for Bayesian methods with varying thresholds. 
The AUCs for the regularization methods are omitted.
\begin{table}[!tb]
	\centering\scriptsize
	\caption{Performance summary for Scenario 1 (high overlapping). Comparison of true positive rate (TPR), false positive rate (FPR), Matthews correlation coefficient (MCC) and area under the ROC curve (AUC). The models compared are the Bayesian joint ESC method proposed in this paper (JESC) \citep{lee2019minimax}, 
		separate ESC method applied for individual group (SESC), multiple PenPC (MPenPC) \citep{liu2019jointDAG}, and joint graphical lasso (JGL) \citep{danaher2014joint}.}
	\scalebox{1}{
		\begin{tabular}{ccccccc}
			\hline
			& Measure & JESC   & SESC   & MPenPC & JGL    & DAGL   \\\hline
			Group 1         & TPR     & 0.8879 & 0.8610  & 0.8924 & 0.9148 & 0.3785 \\
			& FPR     & 0.0045 & 0.0048 & 0.0232 & 0.0365 & 0      \\
			& MCC     & 0.8403 & 0.8193 & 0.6163 & 0.5432 &0.6113  \\
			& AUC     & 0.9761 & 0.9684 & $\cdot$       &  $\cdot$      & $\cdot$       \\ \hline
			Group 2         & TPR     & 0.9148 & 0.9072 & 0.9462 & 0.9372 &0.3668  \\
			& FPR     & 0.0039 & 0.0044 & 0.0211 & 0.0326 & 0      \\
			& MCC     & 0.8664 & 0.8369 & 0.6638 & 0.5769 &0.6009  \\
			& AUC     & 0.9962 & 0.9780  & $\cdot$       &  $\cdot$      & $\cdot$        \\ \hline
			Group 3         & TPR     & 0.8969 & 0.8654 & 0.8879 & 0.8789 &0.3552 \\
			& FPR     & 0.0038 & 0.0041  & 0.0230  & 0.0369 & 0      \\
			& MCC     & 0.8580  & 0.8361 & 0.6152 & 0.5224 &0.5901  \\
			& AUC     & 0.9835 & 0.9805 & $\cdot$       &  $\cdot$      & $\cdot$        \\ \hline
			All edges       & TPR     & 0.8999 & 0.8775 & 0.9088 & 0.9103 &0.3669  \\
			& FPR     & 0.0040  & 0.0044 & 0.0224 & 0.0353 & 0      \\
			& MCC     & 0.8549 & 0.8308 & 0.6317 & 0.5471 &0.6010  \\
			& AUC     & 0.9479 & 0.9289 & $\cdot$       &  $\cdot$      & $\cdot$   \\ \hline
			Differential edges & TPR     & 1      & 0.9050    & 1      & 1      &0.4600    \\
			& FPR     & 0      & 0      & 0      & 0      & 0      \\
			& MCC     & 1      & 0.9098 & 1      & 1      &0.5463     \\
			& AUC     & 1      & 0.9525   & $\cdot$       &  $\cdot$      & $\cdot$      \\\hline
		\end{tabular} \label{table:comp1}
	}
\end{table}

\begin{table}[!tb]
	\centering\scriptsize
	\caption{Performance summary for Scenario 2 (medium overlapping). 
	}
	\scalebox{1}{
		\begin{tabular}{ccccccc}
			\hline
			& Measure & JESC   & SESC   & MPenPC & JGL    & DAGL   \\ \hline
			Group 1         & TPR     & 0.8704   & 0.8475 & 0.9058 & 0.9193 & 0.3565 \\
			& FPR     & 0.0051  & 0.0049 & 0.0227 & 0.0363 & 0    \\
			& MCC     & 0.8195 & 0.8096 & 0.6275 & 0.5465 & 0.5908 \\
			& AUC     & 0.9810  & 0.9836 & $\cdot$       &  $\cdot$      & $\cdot$       \\ \hline
			Group 2         & TPR     & 0.9238 & 0.8654 & 0.9238 & 0.9372 & 0.3632 \\
			& FPR     & 0.0030  & 0.0043 & 0.0216 & 0.0330  & 0      \\
			& MCC     & 0.8901 & 0.8307 & 0.6466 & 0.5748 & 0.5963 \\
			& AUC     & 0.9896 & 0.9885 & $\cdot$       &  $\cdot$      & $\cdot$  \\ \hline
			Group 3         & TPR     & 0.8610  & 0.8834 & 0.8924 & 0.8879 & 0.3529 \\
			& FPR     & 0.0040  & 0.0046 & 0.0232 & 0.0368 & 0 \\
			& MCC     & 0.8335 & 0.8359 & 0.6163 & 0.5276 & 0.5897  \\
			& AUC     & 0.9859 & 0.9853 & $\cdot$       &  $\cdot$      & $\cdot$     \\ \hline
			All edges       & TPR     & 0.8849 & 0.8654 & 0.9073 & 0.9148 & 0.3584 \\
			& FPR     & 0.0041  & 0.0046 & 0.0225 & 0.0354 & 0      \\
			& MCC     & 0.8475 & 0.8254 & 0.6301 & 0.5484 & 0.5930 \\
			& AUC     & 0.9405 & 0.9304 & $\cdot$       &  $\cdot$      & $\cdot$    \\ \hline
			Differential edges & TPR     & 0.8920   & 0.8482   & 0.9241   & 0.9190   & 0.3800    \\
			& FPR     & 0      & 0      & 0.0381   & 0.0814   & 0      \\
			& MCC     & 0.8978 & 0.8584 & 0.8867 & 0.8423   & 0.4835    \\
			& AUC     & 0.9461   & 0.9235   & $\cdot$       &  $\cdot$      & $\cdot$    \\\hline
		\end{tabular} \label{table:comp2}
	}
\end{table}

\begin{table}[!tb]
	\centering\scriptsize
	\caption{Performance summary for Scenario 3 (low overlapping). 
	}
	\scalebox{1}{
		\begin{tabular}{ccccccc}
			\hline
			& Measure & JESC   & SESC   & MPenPC & JGL    & DAGL   \\ \hline
			Group 1         & TPR     & 0.8879 & 0.8520  & 0.8924 & 0.9193 &0.3796  \\
			& FPR     & 0.0042 & 0.0048 & 0.0226 & 0.0360  & 0      \\
			& MCC     & 0.8456 & 0.8140  & 0.6207 & 0.5484 &0.6067  \\
			& AUC     & 0.9866 & 0.9786 & $\cdot$       &  $\cdot$      & $\cdot$      \\ \hline
			Group 2         & TPR     & 0.8969 & 0.8565 & 0.9148 & 0.9193 &0.3330  \\
			& FPR     & 0.0041 & 0.0040  & 0.0236 & 0.0372 & 0      \\
			& MCC     & 0.8526 & 0.8327 & 0.6254 & 0.5422 &0.5705  \\
			& AUC     & 0.9829 & 0.9785 & $\cdot$       &  $\cdot$      & $\cdot$     \\ \hline
			Group 3         & TPR     & 0.8879 & 0.9103 & 0.9148 & 0.9148 &0.3643  \\
			& FPR     & 0.0044 & 0.0047 & 0.0255 & 0.0365 &0  \\
			& MCC     & 0.8403 & 0.8497 & 0.6116 & 0.5432 &0.5967  \\
			& AUC     & 0.9869 & 0.985  & $\cdot$       &  $\cdot$      & $\cdot$     \\ \hline
			All edges       & TPR     & 0.8909 & 0.8729 & 0.9073 & 0.9178 &0.3576  \\
			& FPR     & 0.0042 & 0.0045 & 0.0239 & 0.0366 & 0      \\
			& MCC     & 0.8462 & 0.8321 & 0.6191 & 0.5446 &0.5916  \\
			& AUC     & 0.9433 & 0.9342 & $\cdot$       &  $\cdot$      & $\cdot$     \\ \hline
			Differential edges & TPR     & 0.8530   & 0.8105    & 0.9255  &0.8940   &0.3375   \\
			& FPR     & 0      & 0      & 0.0518   &0.0905  & 0      \\
			& MCC     & 0.8617 & 0.8256 & 0.8753 &0.8392  &0.4459  \\
			& AUC     & 0.9251  & 0.9022   & $\cdot$       &  $\cdot$      & $\cdot$        \\\hline
		\end{tabular} \label{table:comp3}
	}
\end{table}

Based on the simulation results (Tables \ref{table:comp1} to \ref{table:comp3}), we can tell that the proposed method is more conservative in the identification of differential edges compared with frenquentist approaches, as indicated by its lower sensitivity and FPR. The high FPR of the penalized likelihood based methods in selecting differential edges is partly due to the fact that they select a larger number of false positive edges overall and may be because the regularization methods based on cross-validation tend to include many redundant variables resulting in a relatively larger number of  errors compared with those for the Bayesian methods \citep{peterson2020joint}. 

The proposed method achieves the highest MCC in identifying all edges across methods compared and yields a higher AUC compared with the separate inference, especially in the high and medium overlapping scenarios. As indicated in our theoretical results, the estimation performance based on the joint inference benefits the most when all graphs share the common support. 
Figure \ref{fig_1} shows the unnormalized posterior scores in log scale.
Based on Figure \ref{fig_1}, it seems that, in the high and medium overlapping settings, not only does JESC outperform SESC but also the posterior probabilities based on JESC increase faster than SESC during the beginning of the MCMC procedure. 
\begin{figure}[!tb]
	\begin{subfigure}[t]{.3\textwidth}
		\centering
		\includegraphics[width=1.0\linewidth]{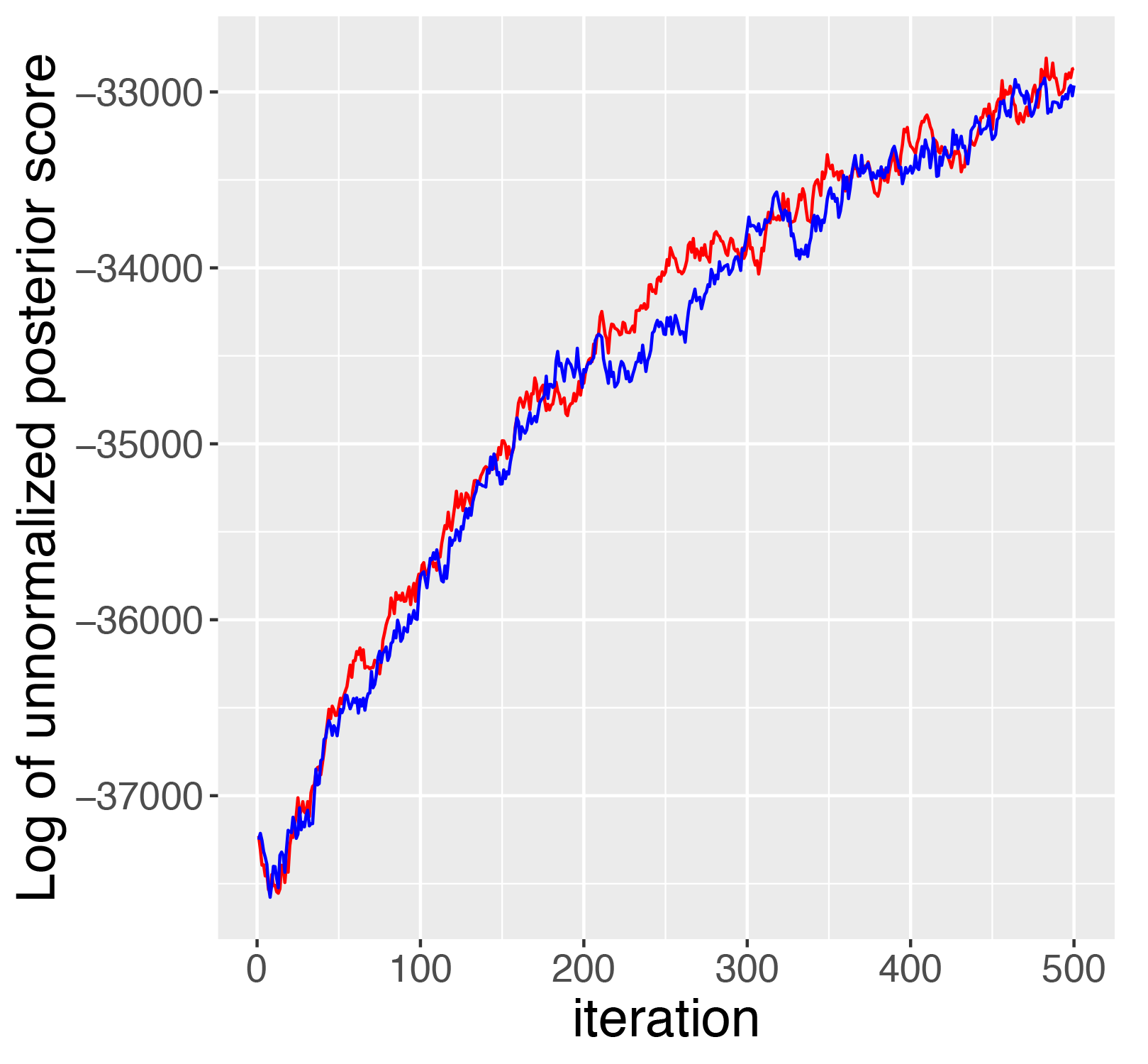}
		\caption{High overlapping}
	\end{subfigure}
	\,\,
	\begin{subfigure}[t]{.3\textwidth}
		\centering
		\includegraphics[width=1.0\linewidth]{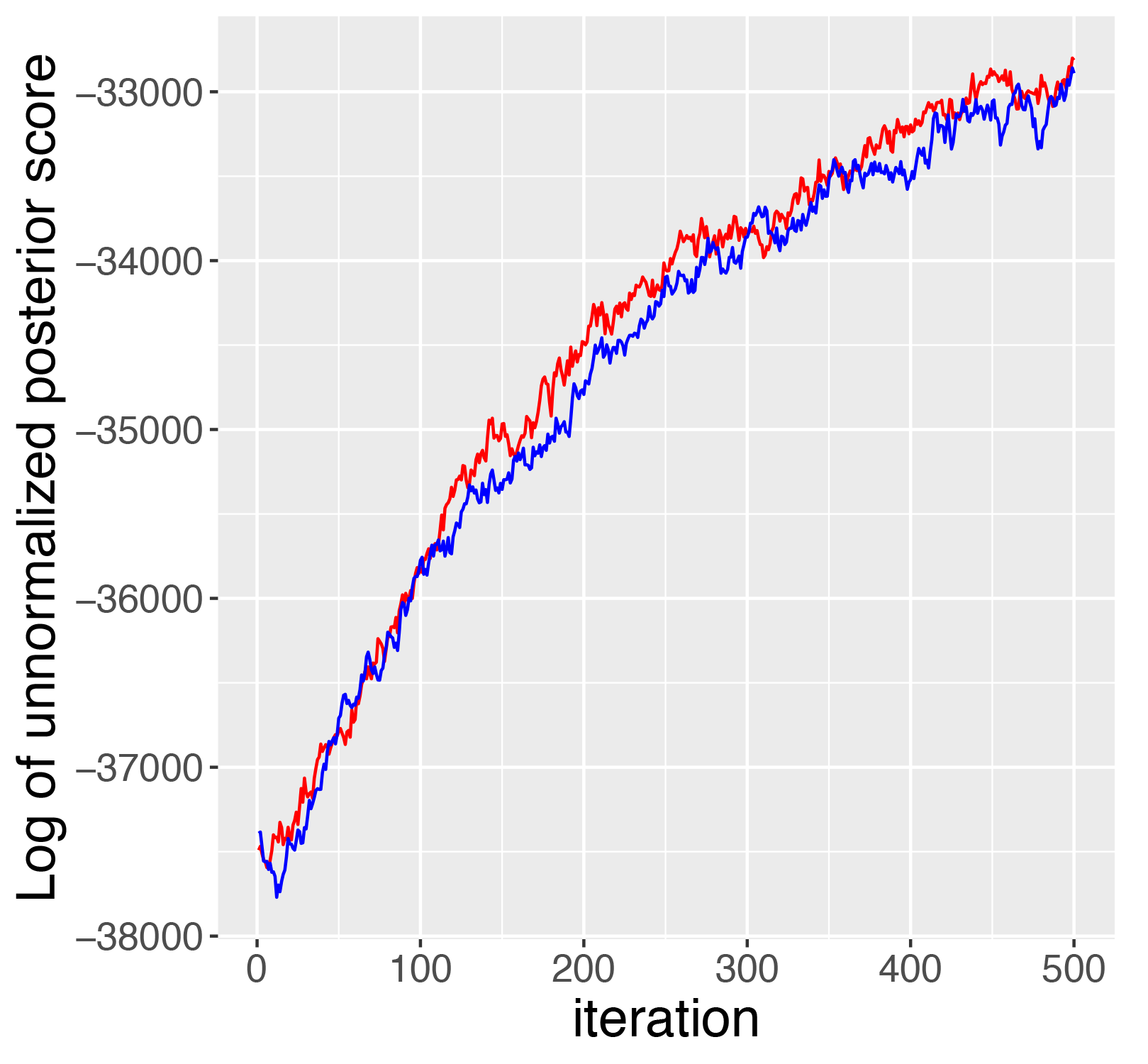}
		\caption{Medium overlapping}
	\end{subfigure}
	\,\,
	\begin{subfigure}[t]{.3\textwidth}
		\centering
		\includegraphics[width=1.21\linewidth]{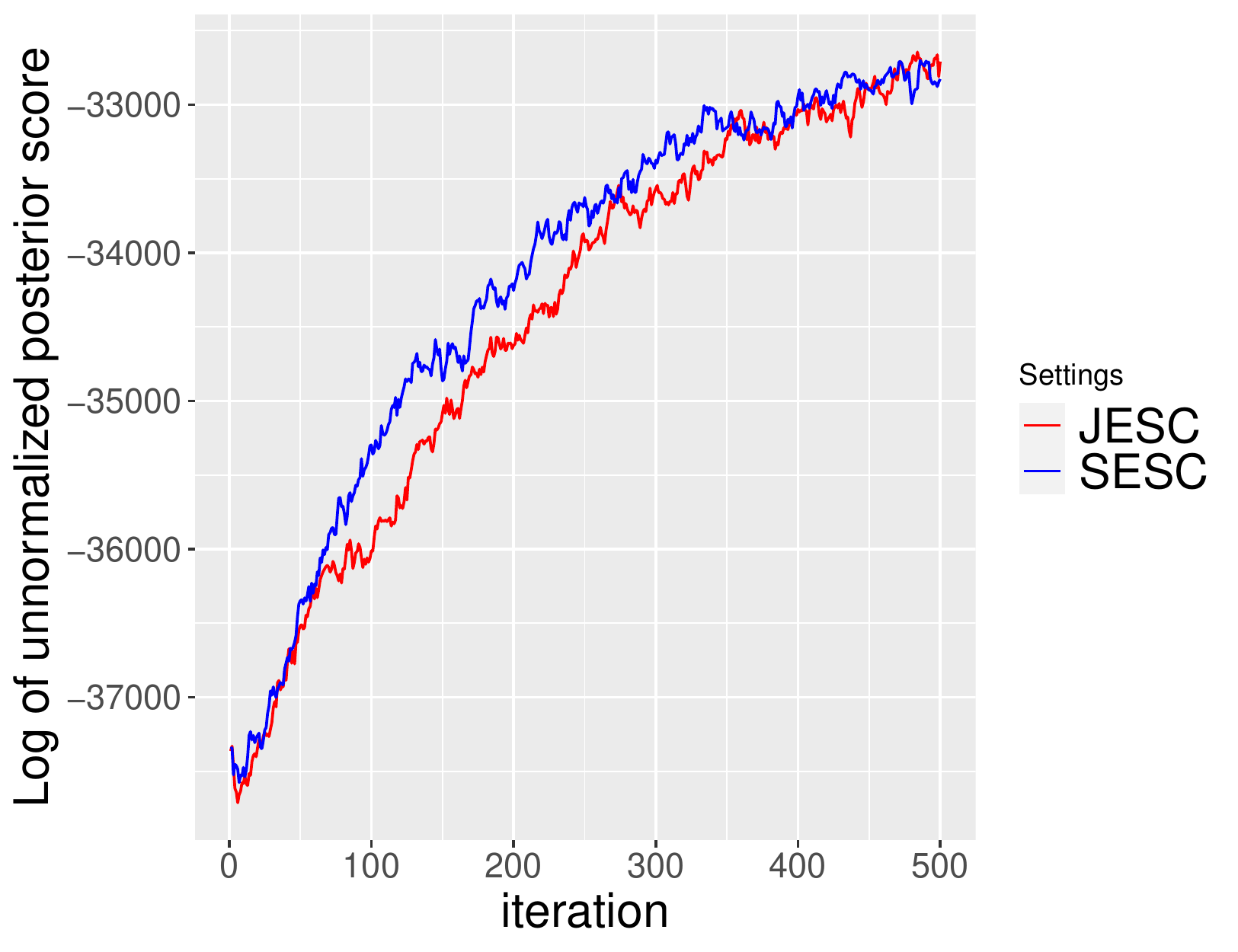}
		\caption{Low overlapping}
	\end{subfigure}
	\caption{The log of unnormalized posterior scores during the first 500 iterations under different scenarios.}
	\label{fig_1}
\end{figure}

\section{Inferring Brain Functional Networks} \label{sec:real}
In this section, we continue the illustration of JESC by applying the proposed method to an fMRI data set for simultaneously inferring multiple brain functional networks.
Parkinson's disease (PD) is a major neurodegenerative disease influenced by both genetic and environmental factors \citep{Halliday:2014}. As the second most common neurodegenerative disorder, PD is characterized by the degeneration of dopamine-producing cells in the brain resulting in motor symptoms and nonmotor features \citep{Mhyre:2012}. Depression is the most common psychiatric symptom in patients with PD, and one of the earliest prodromal comorbidities that can have a significant impact on the quality of life \citep{chagas2013}. Nonmotor features including depression can appear in the earliest phase of the disease even before clinical motor impairment \citep{Lix2010,Shearer2012, Tibar2018}, but the efficacy of medications and psychotherapies for treating depression in PD (DPD) patients remains limited \citep{Abos2017}. Hence, advances in timely detection and concerted management of DPD becomes urgent.  

Up until now, the neural and pathophysiologic mechanisms of DPD remain unclear and are key research priorities for neurologists. A variety of neuroimaging technologies including fMRI, structure MRI, positron emission tomography and electroencephalography have been adopted to study PD. Among these, neuroimaging indicators have achieved considerable progress, and have provided new insights into PD. Resting-state fMRI exploits blood oxygen level-dependent signal to assess the correlation of the networks in different brain areas. An intra- and inter-network functional connectivity study in DPD demonstrated abnormal functional connection in left frontoparietal network, basal ganglia network, salience network and default-mode network \citep{wei2017aberrant}. To understand the underlying functional network changes for both DPD and non-depressed PD (NDPD) patients so that physicians could get an early-diagnosis in time for available treatment, we apply the proposed method to an fMRI data set \citep{wei2017aberrant} for identifying regions of interest that are associated with the aberrant functional network and relevant to the onset of DPD and NDPD.

Twenty-one DPD patients, 49 NDPD patients and 50 matched healthy controls (HC) were recruited. Image data were acquired using a Siemens 3.0-Tesla signal scanner and functional imaging data were collected transversely by using a gradient-recalled echo-planar imaging (GRE-EPI) pulse sequence. We further perform image preprocessing procedure using Data Processing Assistant for Resting-State fMRI (\url{http://rfmri.org/DPARSF}) based on Statistical Parametric Mapping (SPM12, \url{http:// www.fil.ion.ucl.ac.uk/spm/}) operated on the Matlab platform. \cite{Zang:Jiang} proposed the method of Regional Homogeneity (ReHo) to analyze characteristics of regional brain activity and to reflect the temporal homogeneity of neural activity. In particular, we focus on the mReHo maps obtained by dividing the mean ReHo of the whole brain within each voxel in the ReHo map. We further segment the mReHo maps based on the Harvard-Oxford atlas (HOA) and extract all the mReHo signals corresponding to 15 subcortical regions of interest (ROI) (HOA number: 97-112) using the Resting-State fMRI Data Analysis Toolkit. Hence, adapted to our setting, $n_1 = 21$, $n_2 = 49$, $n_3 = 50$, $p = 15$, and the ordering is taken according to the HOA number.

We apply JESC along with other contenders to the resulting mReHo data set consisting of three groups for jointly estimating the functional connectivity networks. The parameter configuration are identical to those in the simulation study. Table \ref{table:real1} lists the number of edges selected by JESC and its competitors. The separate estimation methods (SESC and DAGL) resulted in graphs that share fewer edges in the Cholesky factors for the precision matrices of three groups. JGL resulted in most shared edges, followed by MPenPC and our method (JESC). Overall, JGL and MPenPC selected a lot more linked genes than other methods. JESC and SESC selected less unique edges among the ROIs for DPD than those for NDPD and HC. This might suggest that the patients with DPD lack some important links among the subcortical regions. 
By visualizing the brain connectome as nodes and edges, Figure \ref{fig_2} shows the DAGs for three groups and all the shared edges estimated by JESC. 

Table \ref{table:real2} lists six edges that are unique to the group of DPD identified by JESC. In particular, we discover discriminative connectivity changes between hippocampus and amygdala areas. These findings suggest disease-related alterations of functional connectivity as the basis for faulty information processing in DPD. Our findings are in good agreement with the aberrant functional features in subcortical regions that are related to the onset of DPD as shown in previous studies  \citep{Dan2017,Lin2020,Cao2020}. 

\begin{figure} [htbp]
	\begin{subfigure}[t]{.23\textwidth}
		\centering
		\includegraphics[width=1.0\linewidth]{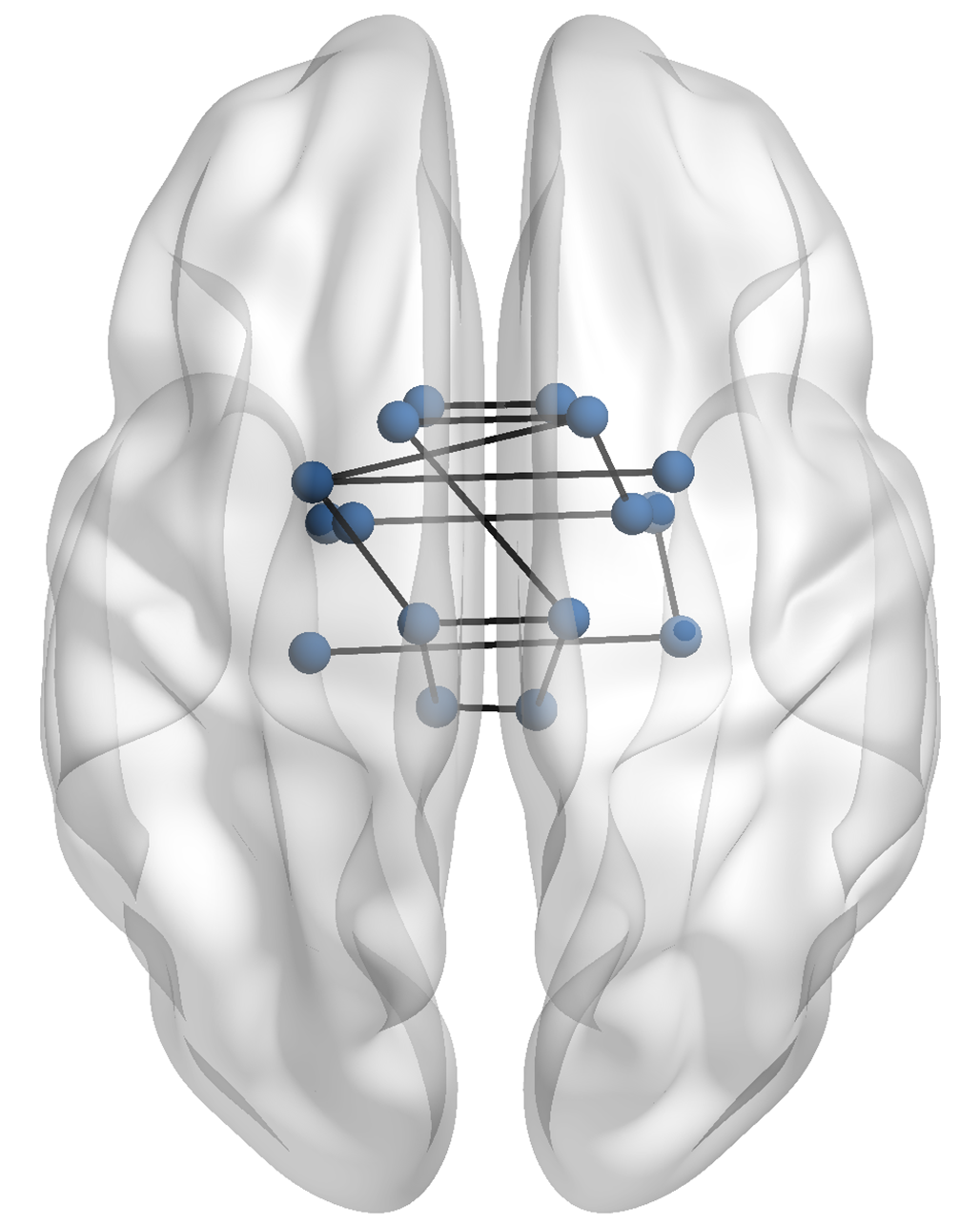}
		\caption{Shared edges}
	\end{subfigure}
	\,\,
	\begin{subfigure}[t]{.23\textwidth}
		\centering
		\includegraphics[width=1.0\linewidth]{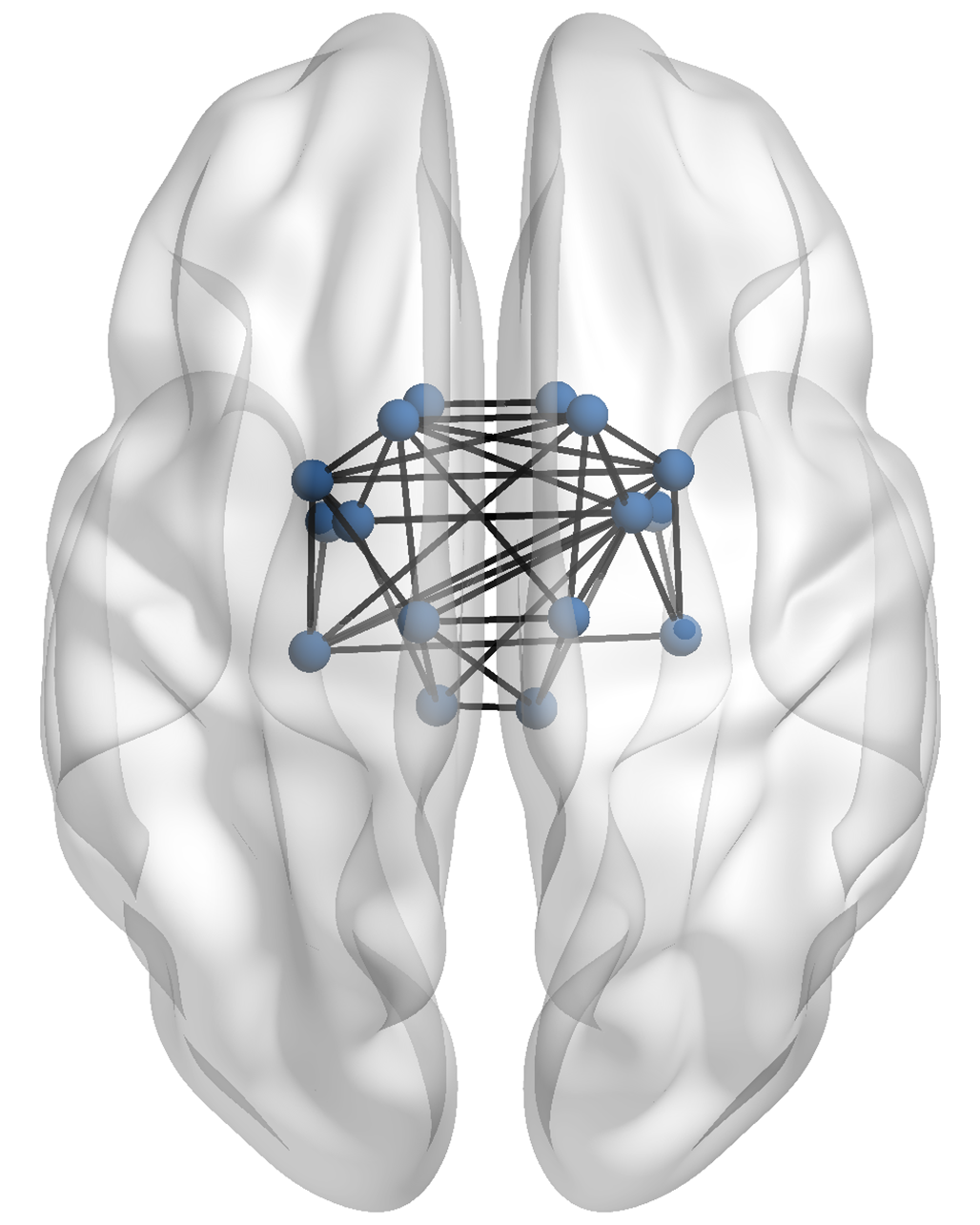}
		\caption{HC}
	\end{subfigure}
	\,\,
	\begin{subfigure}[t]{.23\textwidth}
		\centering
		\includegraphics[width=1.0\linewidth]{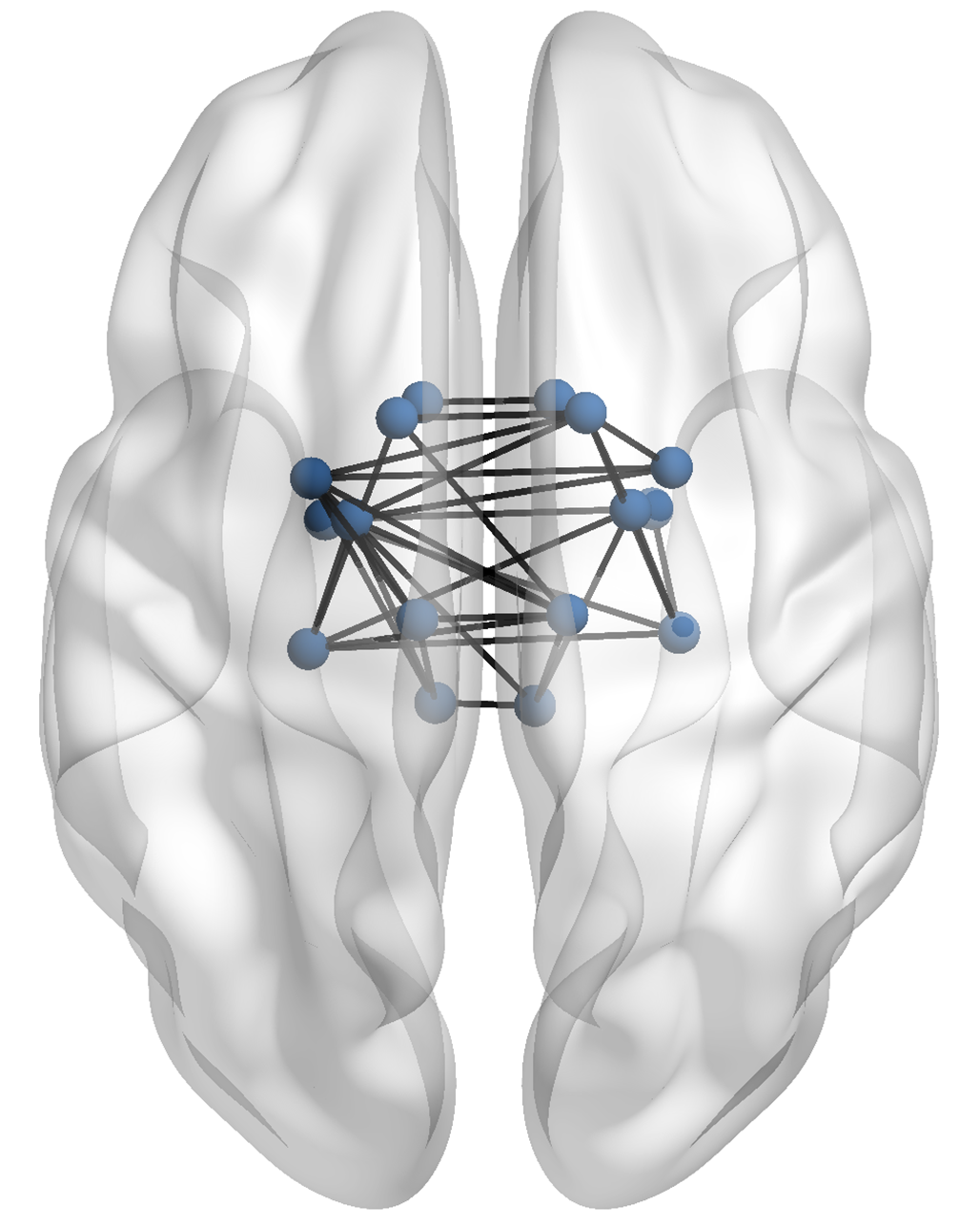}
		\caption{DPD}
	\end{subfigure}
	\,\,
	\begin{subfigure}[t]{.23\textwidth}
		\centering
		\includegraphics[width=1.0\linewidth]{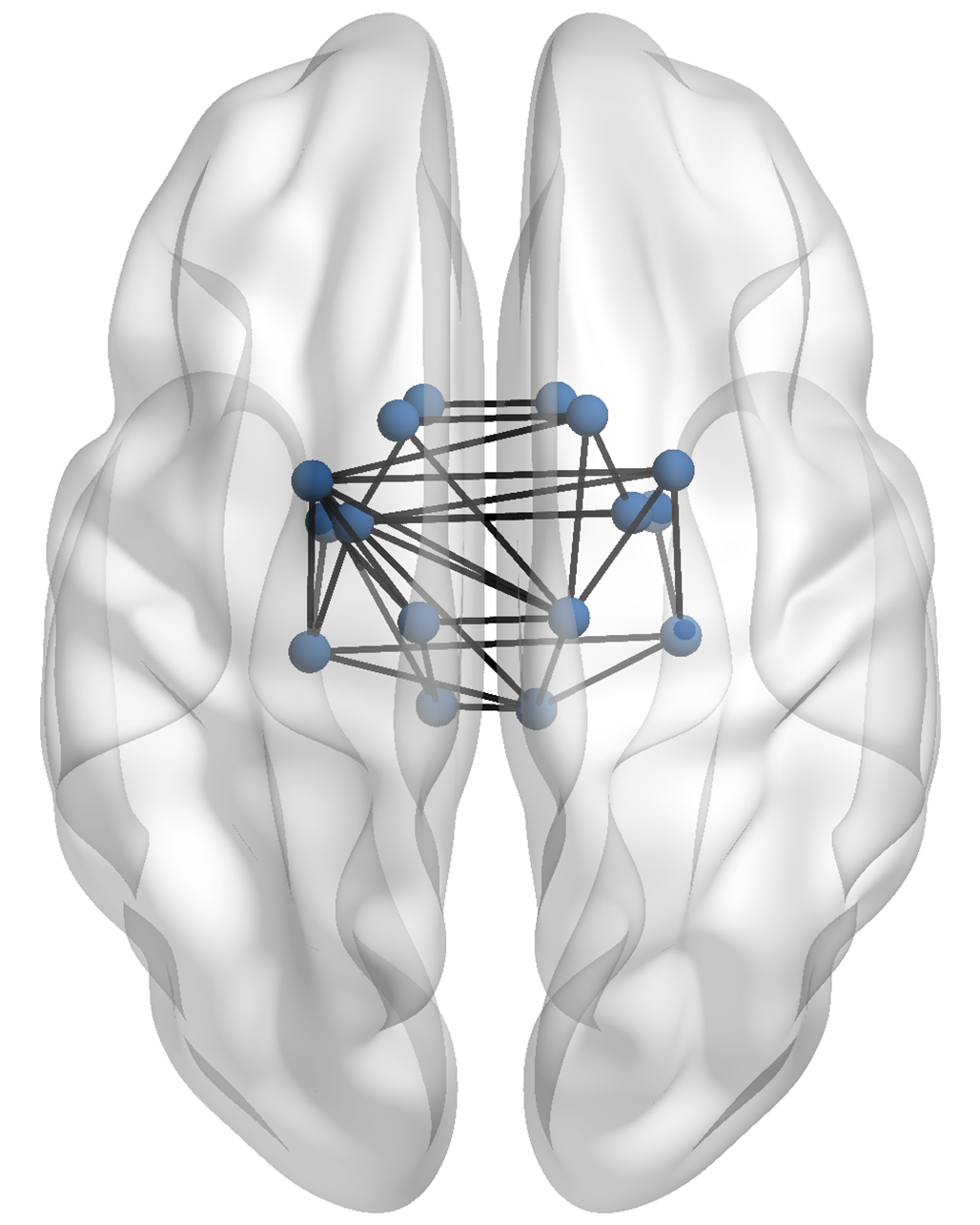}
		\caption{NDPD}
	\end{subfigure}
	\caption{Estimated brain function activity networks for HC, DPD, NDPD and the shared connections among three groups.}
	\label{fig_2}
\end{figure}

\begin{table}[htbp]
	\centering\scriptsize
	\caption{Number of edges selected by the proposed method and its competitors. “DPD unique” counts the number of edges that only appear in the DPD group; “NDPD unique” counts the number of edges that only appear in the NDPD group; “HC unique” counts the number of edges that only appear in the HC group; and “Shared” counts the number of edges shared by all three groups.}
	\begin{tabular}{ccccc}
		\hline
		Method & DPD unique & NDPD unique & HC unique & Shared \\ \hline
		JESC   & 6          & 8           & 12        & 14     \\
		SESC   & 8          & 9           & 11        & 10     \\
		MPenPC & 10         & 5           & 8         & 19     \\
		JGL    & 9          & 3           & 9         & 27     \\
		DAGL   & 5          & 3           & 5         & 5     \\ \hline
	\end{tabular} \label{table:real1}
\end{table}

\begin{table}[htbp]
	\centering\scriptsize
	\caption{Estimated edges that are unique to the group of DPD and the related brain regions indexed in the HOA template. }
	\begin{tabular}{ccccc}
		\hline
		ID & HOA number & Brain region A    & HOA number & Brain region B    \\\hline
		1  & 105        & Left Pallidum     & 99        & Left Thalamus     \\
		2  & 105        & Left Pallidum     & 102        & Right Caudate     \\
		3  & 107        & Left Hippocampus  & 100        & Right Thalamus    \\
		4 & 107        & Left Hippocampus & 101        & Left Caudate      \\
		5 & 108        & Right Hippocampus & 106        & Right Pallidum    \\
		6 & 109        & Left Amygdala     & 108        & Right Hippocampus \\
		\hline
	\end{tabular} \label{table:real2}
\end{table}

	\section{Discussion}\label{sec:disc}
	In this paper, we proposed the JESC prior for Bayesian joint inference of multiple DAGs.
	In high-dimensional settings, the induced posterior attains the joint selection consistency under mild conditions.
	We also showed the advantage of the joint inference over separate inferences, in terms of requiring weaker beta-min conditions, when the DAGs share the common structure.
	The proposed joint inference outperforms other state-of-the-art methods in numerical studies based on simulated data sets. We also applied our method to an fMRI data set, where our results are consistent with previous neurological findings.
	
	Throughout the paper, we focus on the MRF prior to encourage similar structures across all DAGs.
	The other choice of prior can be imposed that depends on the relationship between graphs.
	For example, if there is a natural ordering between $K$ classes so that it is expected that the DAGs were generated based on a Markov chain, one can use a prior,
	\bea
	f( S_{1j}, \ldots, S_{Kj} )  &=& f(S_{1j}) \prod_{k=2}^K \pi(S_{kj} \mid S_{k-1, j}) \\
	&\propto&  \prod_{k=2}^K \exp \Big\{  2 c_2 \sum_{l=1}^{j-1} I( S_{k, jl} = S_{k-1,jl} = 1 )   \Big\}   ,\quad  j=2,\ldots, p 
	\eea
	for some constant $c_2>0$, which encourages similar patterns of sparsity for two consecutive graphs $S_{kj}$ and $S_{k-1,j}$.
	Theoretical properties of the joint inference based on various types of joint priors for $(S_{1j},\ldots, S_{Kj})$, including the above Markov-type prior, may worth investigating as future work.

	\newpage
	\section{Proofs}\label{sec:proof}

	\begin{proof}[Proof of Theorem \ref{thm:selection}]
		Let $S_j = (S_{1j},\ldots, S_{Kj})$ and $S_{0j} = (S_{01,j},\ldots, S_{0K,j})$.
		It suffices to show that
		\bea
		\bbE_0  \Big\{  \pi_\alpha \Big( S_{j} \neq S_{0j } \mid \tilde{\bfX}_n    \Big)   \Big\}  &=& o ( p^{-1}) 
		\eea
		for any $j=2,\ldots ,p$, because
		\bea
		1 - \bbE_0 \Big\{  \pi_\alpha \Big( S_{A_1} = S_{A_{01} }, \ldots, S_{A_K} = S_{A_{0K}}  \mid \tilde{\bfX}_n    \Big)   \Big\} 
		&\le& \sum_{j=2}^p  \bbE_0  \Big\{  \pi_\alpha \Big( S_{j} \neq S_{0j } \mid \tilde{\bfX}_n    \Big)   \Big\}  .
		\eea		
		Note that $\{S_j \neq S_{0j} \}$ is equivalent to $\{ S_{kj} \neq S_{0k,j} \text{ for at least one } k =1,\ldots,K \}$.
		For given $1\le l \le K$ and $1\le k_1 <\ldots < k_l  \le K$, define 
		\bea
		N_{k_1,\ldots, k_l}  &:=&  \Big\{ S_j :    S_{k j} \neq S_{0k,j} \text{ if and only if }  k \in \{  k_1,\ldots, k_l\}   \Big\}   .
		\eea
		Then, we have
		\bean
		&& \pi_\alpha \big( S_{j} \neq S_{0j } \mid \tilde{\bfX}_n    \big) \nonumber  \\
		&=& \sum_{k=1}^K  \sum_{ S_j \in N_k }    \pi_\alpha \big( S_{j} \mid \tilde{\bfX}_n    \big)   
		+ \sum_{k_1 < k_2}  \sum_{ S_j \in N_{k_1, k_2} }     \pi_\alpha \big( S_{j} \mid \tilde{\bfX}_n    \big) + \sum_{k_1 < k_2 <k_3}  \sum_{ S_j \in N_{k_1, k_2, k_3} }     \pi_\alpha \big( S_{j} \mid \tilde{\bfX}_n    \big)  \nonumber \\
		&+& \cdots \,\, + \sum_{S_j \in N_{1,\ldots, K} }     \pi_\alpha \big( S_{j} \mid \tilde{\bfX}_n    \big)   . \label{post_neq}
		\eean
		
		The first term in \eqref{post_neq} can be divided into two parts:
		\bea
		&& \sum_{k=1}^K  \bbE_0 \big\{ \pi_\alpha (S_j \in N_k \mid \tilde{\bfX}_n)  \big\}  \\ 
		&=& \sum_{k=1}^K \Big[  \bbE_0 \big\{ \pi_\alpha (S_j \in N_k, S_{kj} \supsetneq S_{0k,j} \mid \tilde{\bfX}_n)  \big\} + \bbE_0 \big\{ \pi_\alpha (S_j \in N_k , S_{kj} \nsupseteq S_{0k,j}  \mid \tilde{\bfX}_n)  \big\} \Big] .
		\eea
		Let $\pi^I( S_{kj } \mid \bfX_{n_k}) \propto f_\alpha (\bfX_{n_k} \mid S_{kj}) \pi(S_{kj}) $ be the posterior for $S_{kj}$ based on the separate inference for each DAG.
		Note that if $S_j \in N_{k}$, then we have
		\bea
		\frac{\pi_\alpha \big( S_{j} \mid \tilde{\bfX}_n    \big)}{\pi_\alpha \big( S_{0j} \mid \tilde{\bfX}_n    \big)} 
		&=& \frac{\pi_\alpha^I (S_{k j} \mid \bfX_{n_{k}}) }{\pi_\alpha^I (S_{0k,j} \mid \bfX_{n_{k}}) }   \frac{f(S_{1j},\ldots, S_{Kj})}{f(S_{01,j},\ldots, S_{0K,j} )}  
		\eea
		and
		\bea
		\frac{f(S_{1j},\ldots, S_{Kj})}{f(S_{01,j},\ldots, S_{0K,j} )}  
		&=& \exp \Big[  c_{2j} \sum_{k' \neq k} \big\{ |S_{kj} \cap S_{0k',j}|  - |S_{0k,j} \cap S_{0k',j}|    \big\}  \Big] \\
		&\le& \exp \Big[  c_{2j} \sum_{k' \neq k} \big\{ |S_{kj} \cap S_{0k',j}|   \big\}  \Big] \\
		&\le& \exp \big\{ c_{2j} (K-1) (j-1)  \big\} \,\, \le \,\, \exp \big\{c_{2j}(j-1) K  \big\}  .
		\eea
		Then, we have
		\bea
		&& \sum_{k=1}^K \Big[  \bbE_0 \big\{ \pi_\alpha (S_j \in N_k, S_{kj} \supsetneq S_{0k,j} \mid \tilde{\bfX}_n)  \big\} \\
		&\le& \sum_{k=1}^K \sum_{S_j \in N_k, S_{kj} \supsetneq S_{0k,j}}   \bbE_0 \Big\{ \frac{\pi_\alpha ( S_j \mid \tilde{\bfX}_n)}{\pi_\alpha ( S_{0j} \mid \tilde{\bfX}_n)}  \Big\}  \\
		&=& \sum_{k=1}^K \sum_{S_j \in N_k, S_{kj} \supsetneq S_{0k,j}}   \bbE_0 \Big\{ \prod_{k'=1}^K \frac{\pi_\alpha^I ( S_{k'j} \mid \bfX_{n_{k'}})}{\pi_\alpha^I ( S_{0k',j} \mid \bfX_{n_{k'}})}  \Big\}  \frac{f(S_{01,j},\ldots, S_{0k-1,j}, S_{kj}, S_{0k+1,j},\ldots, S_{0K,j} )}{f(S_{01,j},\ldots, S_{0K,j} )} \\
		&\le& \sum_{k=1}^K \sum_{ S_{kj} \supsetneq S_{0k,j}}   \bbE_0 \Big\{ \frac{\pi_\alpha^I ( S_{kj} \mid \bfX_{n_{k}})}{\pi_\alpha^I ( S_{0k,j} \mid \bfX_{n_{k}})}  \Big\}  \exp \big\{c_{2j}(j-1) K  \big\}  \\
		&\lesssim& K p^{-c_1} R_j  \exp \big\{c_{2j}(j-1) K  \big\}  \\
		&\le&  K p^{- (c_1-1)} \exp \big\{c_{2j}(j-1) K  \big\}  \,\,\le\,\, K p^{ -\{(C_{\rm bm}-c_1-1)\wedge (c_1-1)\} }  \exp \big\{c_{2j}(j-1) K  \big\} 
		\eea		
		where the last inequality follows from the proof of Lemma 6.1 in \cite{lee2019minimax}.
		Let $N_{S_{kj}, \alpha, \chi^2}$ be the set defined in the proof of Theorem 3.1 in \cite{lee2019minimax}.
		Then,
		\bea
		&&\sum_{k=1}^K \bbE_0 \big\{ \pi_\alpha (S_j \in N_k , S_{kj} \nsupseteq S_{0k,j}  \mid \tilde{\bfX}_n)  \big\} \Big] \\
		&\le& \sum_{k=1}^K \sum_{S_{kj} \nsupseteq S_{0k,j} } \bbP_0 \big( \bfX_{n_k} \in N_{S_{kj}, \alpha, \chi^2} \big)
		\\
		&&+ \sum_{k=1}^K \sum_{S_{kj} \nsupseteq S_{0k,j} } \bbE_0 \Big\{ \frac{\pi_\alpha^I ( S_{kj} \mid \bfX_{n_{k}})}{\pi_\alpha^I ( S_{0k,j} \mid \bfX_{n_{k}})} I(\bfX_{n_k} \in N^c_{S_{kj}, \alpha, \chi^2} ) \Big\} \exp \big\{c_{2j}(j-1) K  \big\}  \\
		&\lesssim& \sum_{k=1}^K \sum_{S_{kj} \nsupseteq S_{0k,j} }\exp \Big\{  - \frac{(\epsilon')^2 \epsilon_0^2 }{64(1+2\epsilon_0)^2 } n_k \Big\}  \\
		&&+ \sum_{k=1}^K \sum_{S_{kj} \nsupseteq S_{0k,j} } \bbE_0 \Big\{ \frac{\pi_\alpha^I ( S_{kj} \mid \bfX_{n_{k}})}{\pi_\alpha^I ( S_{0k,j} \mid \bfX_{n_{k}})} I(\bfX_{n_k} \in N^c_{S_{kj}, \alpha, \chi^2} ) \Big\} \exp \big\{c_{2j}(j-1) K  \big\}  \\
		&\lesssim&  K \exp \Big\{  - \frac{(\epsilon')^2 \epsilon_0^2 }{128(1+2\epsilon_0)^2 } \min_k n_k \Big\}  +K \big(  p^{-C_{\rm bm} +1 } R_j +  p^{-C_{\rm bm} + c_1 +1}  \big) \exp \big\{c_{2j}(j-1) K  \big\}  \\
		&\lesssim& K p^{ -(C_{\rm bm}-c_1-1) } \exp \big\{c_{2j}(j-1) K  \big\}  \\
		&\le&  K p^{ -\{(C_{\rm bm}-c_1-1)\wedge (c_1-1)\} }  \exp \big\{c_{2j}(j-1) K  \big\} 
		\eea
		where the second and third inequalities follow from Lemma 6.2 and the proof of Theorem 3.1 in \cite{lee2019minimax}.
		The last inequality holds by Condition (P) because we assume $s_0 \ge C_{\rm bm} - c_1 -1$.

		Now consider the second term in \eqref{post_neq}.
		Note that if $S_j \in N_{k_1, k_2}$, then we have
		\bea
		\frac{\pi_\alpha \big( S_{j} \mid \tilde{\bfX}_n    \big)}{\pi_\alpha \big( S_{0j} \mid \tilde{\bfX}_n    \big)} 
		&=& \frac{\pi_\alpha^I (S_{k_1 j} \mid \bfX_{n_{k_1}}) }{\pi_\alpha^I (S_{0k_1,j} \mid \bfX_{n_{k_1}}) } \frac{\pi_\alpha^I (S_{k_2 j} \mid \bfX_{n_{k_2}}) }{\pi_\alpha^I (S_{0k_2,j} \mid \bfX_{n_{k_2}}) } \frac{f(S_{1j},\ldots, S_{Kj})}{f(S_{01,j},\ldots, S_{0K,j} )}  
		\eea
		and
		\bea
		\frac{f(S_{1j},\ldots, S_{Kj})}{f(S_{01,j},\ldots, S_{0K,j} )}  
		&\le& \exp \Big[ c_{2j} \big\{  |S_{k_1j}\cap S_{k_2 j}| + \sum_{k' \notin \{k_1, k_2\} }| S_{k_l j} \cap S_{0k' ,j}  |  \big\}   \Big] \\
		&\le& \exp \Big[  c_{2j} \big\{ j-1  + 2 (K-2) (j-1)  \big\}  \Big]  \\
		&\le& \exp \big\{ c_{2j} (j-1) 2 K  \big\} .
		\eea
		If $S_j \in N_{k_1, k_2}$, then one of the followings holds: (1) $S_{k_1 j} \supsetneq S_{0k_1,j}$ and $S_{k_2 j}\supsetneq S_{0k_2,j}$, (2) $S_{k_1 j} \supsetneq S_{0k_1,j}$ and $S_{k_2 j}\nsupseteq S_{0k_2,j}$, (3) $S_{k_1 j} \nsupseteq S_{0k_1,j}$ and $S_{k_2 j}\supsetneq S_{0k_2,j}$ or (4) $S_{k_1 j} \nsupseteq S_{0k_1,j}$ and $S_{k_2 j}\nsupseteq S_{0k_2,j}$.
		For example, by the similar arguments used in the previous paragraph,
		\bea
		&&  \sum_{k_1 < k_2} \sum_{S_j \in N_{k_1, k_2}, \atop S_{k_1 j} \supsetneq S_{0k_1,j} ,  S_{k_2 j}\nsupseteq S_{0k_2,j} } \bbE_0 \big\{ \pi_\alpha ( S_j \mid \tilde{\bfX}_n )  \big\}  \\
		&\lesssim& \sum_{k_1<k_2} \sum_{S_{k_2 j} \nsupseteq S_{0k_2,j} } \bbP_0 \big( \bfX_{n_{k_2}} \in N_{S_{k_2j}, \alpha, \chi^2}  \big) \\
		&+&  \sum_{k_1<k_2} \sum_{S_{k_1 j} \supsetneq S_{0k_1,j} } \bbE_0 \Big\{ \frac{\pi_\alpha^I (S_{k_1 j} \mid \bfX_{n_{k_1}}) }{\pi_\alpha^I (S_{0k_1,j} \mid \bfX_{n_{k_1}}) }  \Big\}    \sum_{S_{k_2 j}\nsupseteq S_{0k_2,j}  } \bbE_0 \Big\{   \frac{\pi_\alpha^I (S_{k_2 j} \mid \bfX_{n_{k_2}}) }{\pi_\alpha^I (S_{0k_2,j} \mid \bfX_{n_{k_2}}) }  I(\bfX_{n_{k_2}} \in N^c_{S_{k_2j}, \alpha, \chi^2})  \Big\} \\
		&&   \times \exp \big\{ c_{2j} (j-1) 2 K  \big\}   \\
		&\lesssim&  \sum_{k_1 < k_2} p^{ -2 \{ (C_{\rm bm} -c_1-1) \wedge (c_1-1) \}  }   \exp \big\{ c_{2j} (j-1) 2 K  \big\} \\
		&\le&  K^2 p^{ -2 \{ (C_{\rm bm} -c_1-1) \wedge (c_1-1) \}  }   \exp \big\{ c_{2j} (j-1) 2 K  \big\}   .
		\eea
		Thus, by applying the similar arguments to the above four cases, the expectation of the second term in \eqref{post_neq} is 
		\bea
		&& \sum_{k_1 < k_2}  \sum_{ S_j \in N_{k_1, k_2} }   \bbE_0 \big\{   \pi_\alpha \big( S_{j} \mid \tilde{\bfX}_n    \big) \big\}  \\
		&\le& \sum_{k_1 < k_2} \sum_{ S_j \in N_{k_1, k_2} }   \bbE_0 \Big\{   \frac{\pi_\alpha^I (S_{ j} \mid \bfX_{n}) }{\pi_\alpha^I (S_{0,j} \mid \bfX_{n}) }   \Big\}    \frac{f( S_j )}{f(S_{0j})}   \\
		&\le& \sum_{k_1 < k_2} \sum_{ S_{k_1 j} \neq S_{0k_1,j} }  \sum_{ S_{k_2 j} \neq S_{0k_2,j} } \bbE_0 \Big\{   \frac{\pi_\alpha^I (S_{k_1 j} \mid \bfX_{n_{k_1}}) }{\pi_\alpha^I (S_{0k_1,j} \mid \bfX_{n_{k_1}}) } \frac{\pi_\alpha^I (S_{k_2 j} \mid \bfX_{n_{k_2}}) }{\pi_\alpha^I (S_{0k_2,j} \mid \bfX_{n_{k_2}}) }   \Big\}    \exp \big\{ c_{2j} (j-1) 2 K  \big\}  \\
		&\lesssim& K^2 p^{ -2 \{ (C_{\rm bm} -c_1-1) \wedge (c_1-1) \}  }   \exp \big\{ c_{2j} (j-1) 2 K  \big\}  .
		\eea

		Note that if $S_j \in N_{k_1,\ldots, k_l}$, then we have
		\bea
		\frac{f(S_{1j},\ldots, S_{Kj})}{f(S_{01,j},\ldots, S_{0K,j} )}  
		&\le& \exp \Big[ c_{2j}\Big\{ \frac{l(l-1)}{2} (j-1)  + l (K-l) (j-1) \Big\}  \Big] \\
		&\le& \exp \big\{ c_{2j} (j-1) \, l \, K  \big\} .
		\eea
		Therefore,  by repeatedly applying the similar arguments, we have
		\bea
		\bbE_0  \Big\{  \pi_\alpha \Big( S_{j} \neq S_{0j } \mid \tilde{\bfX}_n    \Big)   \Big\}  
		&\lesssim& \sum_{k=1}^K   \Big[  \frac{ K \exp \big\{ c_{2j} (j-1) K  \big\}   }{ p^{ \{ (C_{\rm bm} -c_1-1) \wedge (c_1-1) \}  }  }   \Big]^k \\
		&\le& \sum_{k=1}^K   \Big[  \frac{ K e^K  }{ p^{ \{ (C_{\rm bm} -c_1-1) \wedge (c_1-1) \}  }  }   \Big]^k  \\
		&\lesssim&  \frac{ K e^K  }{ p^{ \{ (C_{\rm bm} -c_1-1) \wedge (c_1-1) \}  }  }  \,\, =\,\, o(p^{-1})  ,
		\eea
		because we assume $C_{\rm bm} > c_1 +2$, $c_1 >2$, $c_{2j} \le 1/(j-1)$ and $K=o(\log p)$. 	
	\end{proof}

	\begin{proof}[Proof of Theorem \ref{thm:advan1}]
		We will only show that \eqref{cond_prob_inc} holds for $k=1$ when $\cup_{k' =2}^K S_{0k',j} \subseteq S_{01,j}$, but one can easily check the other cases using similar arguments.
		Let $S_j = (S_{1j},\ldots, S_{Kj})$ and $S_{0j} = (S_{01,j},\ldots, S_{0K,j})$.
		Note that 
		\bea
		\frac{\pi_\alpha(S_j \mid \tilde{\bfX}_n) }{\pi_\alpha(S_{0j} \mid \tilde{\bfX}_n ) } 
		&=& \frac{\pi_\alpha^I (S_j \mid \tilde{\bfX}_n) }{\pi_\alpha^I (S_{0j} \mid \tilde{\bfX}_n) }  \frac{f(S_j) }{f(S_{0j})} .
		\eea
		Because we assume $\cup_{k' =2}^K S_{0k',j} \subseteq S_{01,j}$, it holds that
		$f( S_{1,j}, S_{02,j},\ldots, S_{0K, j} ) \le f( S_{01,j}, S_{02,j} ,\ldots, S_{0K, j} )$ for any $S_{1j} \neq S_{01,j}$.		
		Thus, 
		\bea
		\frac{\pi_\alpha(S_{1j} , S_{02, j},\ldots, S_{0K,j} \mid \tilde{\bfX}_n) }{\pi_\alpha(S_{01,j}, S_{02,j},\ldots, S_{0K,j} \mid \tilde{\bfX}_n ) } 
		&=& \frac{\pi_\alpha^I (S_{1j} \mid \tilde{\bfX}_n) }{\pi_\alpha^I (S_{01,j} \mid \tilde{\bfX}_n) }  \frac{f( S_{1j}, S_{02,j},\ldots, S_{0K,j} ) }{f(S_{01,j}, S_{02,j},\ldots, S_{0K,j} )}  \\
		&\le& \frac{\pi_\alpha^I (S_{1j} \mid \tilde{\bfX}_n) }{\pi_\alpha^I (S_{01,j} \mid \tilde{\bfX}_n) }  
		\eea
		for any $S_{1j} \neq S_{01, j}$.
		Then, we have
		\bea
		\frac{1 - \pi_\alpha(S_{01,j} \mid S_{02, j},\ldots, S_{0K,j} , \tilde{\bfX}_n) }{\pi_\alpha(S_{01,j} \mid S_{02,j},\ldots, S_{0K,j} \tilde{\bfX}_n ) }  
		&=& \sum_{S_{1j} \neq S_{01, j} } \frac{\pi_\alpha(S_{1j} \mid S_{02, j},\ldots, S_{0K,j} , \tilde{\bfX}_n) }{\pi_\alpha(S_{01,j} \mid S_{02,j},\ldots, S_{0K,j} \tilde{\bfX}_n ) }   \\
		&=& \sum_{S_{1j} \neq S_{01, j} } \frac{\pi_\alpha(S_{1j} , S_{02, j},\ldots, S_{0K,j} \mid \tilde{\bfX}_n) }{\pi_\alpha(S_{01,j}, S_{02,j},\ldots, S_{0K,j} \mid \tilde{\bfX}_n ) }   \\
		&\le& \sum_{S_{1j} \neq S_{01, j} }  \frac{\pi_\alpha^I (S_{1j} \mid \tilde{\bfX}_n) }{\pi_\alpha^I (S_{01,j} \mid \tilde{\bfX}_n) }  = \frac{1 - \pi_\alpha^I (S_{01,j} \mid \tilde{\bfX}_n) }{\pi_\alpha^I (S_{01,j} \mid \tilde{\bfX}_n) }   ,
		\eea
		which implies 
		\bea
		\pi_\alpha \big(  S_{01, j} \mid S_{01,j},\ldots, S_{0 2, j} , \ldots, S_{0K, j} , \tilde{\bfX}_n  \big)   
		&\ge& \pi_\alpha^I ( S_{01, j}  \mid \bfX_{n_k} )   .
		\eea
	\end{proof}

	\begin{proof}[Proof of Theorem \ref{thm:advan2}]
		In this proof, let $S_j = (S_{1j}, \ldots, S_{j-1 j})$ and $S_{0j}=(S_{0,1j},\ldots, S_{0, j-1 j})$ be the (common) support of the $j$th row of the lower triangular part of $S_A$ and $S_0$, respectively.
		Let 
		\bea
		\tilde{\pi}_\alpha ( S_{A} \mid \tilde{\bfX}_n )  &\propto& \pi_\alpha ( S_{A_1}= \cdots =S_{A_K} = S_{A} \mid \tilde{\bfX}_n ) 
		\eea
		be the joint posterior for $(S_{A_1}, \ldots, S_{A_K})$ restricted to common supports.
		Then,
		\bean
		\tilde{\pi}_\alpha ( S_{A} \neq S_0 \mid \tilde{\bfX}_n ) 
		&\le& \sum_{j=2}^p \tilde{\pi}_\alpha ( S_{j} \neq S_{0j} \mid \tilde{\bfX}_n )\nonumber \\
		&=& \sum_{j=2}^p \tilde{\pi}_\alpha ( S_{j} \supsetneq S_{0j} \mid \tilde{\bfX}_n )  + \sum_{j=2}^p \tilde{\pi}_\alpha ( S_{j} \nsupseteq S_{0j} \mid \tilde{\bfX}_n ) . \label{two_common}
		\eean
		
		The expectation of the first part of \eqref{two_common} is bounded above by
		\bea
		\sum_{j=2}^p \bbE_0 \Big\{  \tilde{\pi}_\alpha ( S_{j} \supsetneq S_{0j} \mid \tilde{\bfX}_n )  \Big\} 
		&\le& \sum_{j=2}^p \sum_{S_{j} \supsetneq S_{0j} }  \bbE_0 \Big\{  \frac{\tilde{\pi}_\alpha ( S_j \mid \tilde{\bfX}_n )}{\tilde{\pi}_\alpha (  S_{0j} \mid \tilde{\bfX}_n )}  \Big\} \\
		&=&  \sum_{j=2}^p \sum_{S_{j} \supsetneq S_{0j} }  \bbE_0 \Big\{ \prod_{k=1}^K \frac{\tilde{\pi}_\alpha^I ( S_{kj} = S_j \mid \bfX_{n_k} )}{\tilde{\pi}_\alpha^I (  S_{kj} = S_{0j} \mid \bfX_{n_k} )}  \Big\}\frac{{f}(S_j,\ldots, S_j )}{{f}(S_{0j},\ldots, S_{0j} )}  \\
		&\lesssim&  \sum_{j=2}^p \sum_{S_{j} \supsetneq S_{0j} }    \Big\{ \frac{{\pi}(S_j) }{{\pi}(S_{0j})} \Big\}^{ K} c_{\alpha, \gamma}^{K(|S_j|- |S_{0j}| ) } \exp \big\{ c_{2j} K (K-1)  (|S_j| - |S_{0j}| )  \big\}  \\
		&\lesssim&  \sum_{j=2}^p  c_{\alpha,\gamma}^{K  }  p^{-c_1  K }  R_j \exp \big\{ c_{2j} K (K-1)   ( j-1 )  \big\}  \\
		&\lesssim& \exp \Big\{  - c_1  K \log p + \log p + K \log c_{\alpha,\gamma} + K (K-1) \Big\}  \,\,=\,\, o(1) ,
		\eea
		where $c_{\alpha,\gamma} = (1+ \alpha/\gamma)^{-1/2}\{2/(1-\alpha) \}^{1/2}$, by the proof of Lemma 6.1 in \cite{lee2019minimax}, $c_1>1$, $c_{2j} \le 1/(j-1)$ and $K  = o(\log p)$.
		
		On the other hand, the expectation of the second part of \eqref{two_common} is 
		\bean
		&& \sum_{j=2}^p \bbE_0 \Big\{  \tilde{\pi}_\alpha ( S_{j} \nsupseteq S_{0j} \mid \tilde{\bfX}_n )  \Big\} \nonumber \\
		&\le& \sum_{j=2}^p \sum_{S_{j} \nsupseteq S_{0j} } \sum_{k=1}^K \bbP_0 \big( \bfX_{n_k} \in N_{S_{j}, \alpha, \chi^2 } \big) \nonumber\\
		&&+ \sum_{j=2}^p \sum_{S_{j} \nsupseteq S_{0j} } \bbE_0 \Big\{ \tilde{\pi}_\alpha (S_j \nsupseteq S_{0j} \mid \tilde{\bfX}_n ) I( \bfX_{n_k} \in N^c_{S_{j}, \alpha, \chi^2 } , \forall k )  \Big\}  \nonumber\\
		&\lesssim& p K \exp \Big\{  - \frac{(\epsilon ')^2 \epsilon_0^2 }{128(1+2\epsilon_0)^2 }  \min_k n_k \Big\} \label{pk_o1}\\
		&&+ \sum_{j=2}^p \sum_{S_{j} \nsupseteq S_{0j} } \bbE_0 \Big\{ \prod_{k=1}^K  \frac{ \tilde{\pi}_\alpha^I (S_{kj}= S_j \mid \bfX_{n_k} ) }{\tilde{\pi}_\alpha^I (S_{kj}= S_{0j} \mid \bfX_{n_k}) } I( \bfX_{n_k} \in N^c_{S_{j}, \alpha, \chi^2 }   )  \Big\} \frac{{f}(S_j,\ldots, S_j )}{{f}(S_{0j},\ldots, S_{0j} )} . \label{rem_o1}
		\eean
		Note that \eqref{pk_o1} is of order $o(1)$ and
		\bea
		\frac{{f}(S_j,\ldots, S_j )}{{f}(S_{0j},\ldots, S_{0j} )}
		&\le& \exp \big\{ c_{2j} K (K-1) \big| |S_j|- |S_{0j}|  \big|  \big\} \,\,\le\,\, \exp \big\{  K(K-1)  \big\}  .
		\eea
		By the proof of Theorem 3.1 in \cite{lee2019minimax}, 
		\bea
		&&  \bbE_0 \Big\{ \prod_{k=1}^K \frac{\tilde{\pi}_\alpha^I (S_{kj}= S_j \mid \bfX_{n_k} ) }{\tilde{\pi}_\alpha^I (S_{kj}= S_{0j} \mid \bfX_{n_k}) } I( \bfX_{n_k} \in N^c_{S_{j}, \alpha, \chi^2 }   )  \Big\}  \\
		&\le&  \prod_{k=1}^K \frac{{\pi}(S_j) }{{\pi}(S_{0j}) } \nu_1^{|S_{0j}| - |S_j|} \nu_2^{|S_j|-|S_{0j}\cap S_j|} \exp \Big\{ - \frac{\alpha(1-\alpha)}{4}\frac{\epsilon_0^2(1-2\epsilon_0)^2}{4} \, n_k \| a_{0k, S_{0j}\cap S_{j}^c} \|_2^2  \Big\}   \\
		&& +  \Big\{ \frac{{\pi}(S_j) }{{\pi}(S_{0j}) } \nu_1^{|S_{0j}| - |S_j|} \nu_2^{|S_j|-|S_{0j}\cap S_j|}  \Big\}^K \sum_{k=1}^K \bbP_0( \bfX_{n_k}\in N_{j, S_{kj}})    \\
		&\le&   \Big\{ \frac{{\pi}(S_j) }{{\pi}(S_{0j}) } \nu_1^{|S_{0j}| - |S_j|} \nu_2^{|S_j|-|S_{0j}\cap S_j|}  \Big\}^K    \exp \Big\{ - (|S_{0j}| - |S_{j} \cap S_{0j}|  )  C_{\rm bm}   K  \log p  \Big\}   \\
		&&  +  \Big\{ \frac{{\pi}(S_j) }{{\pi}(S_{0j}) } \nu_1^{|S_{0j}| - |S_j|} \nu_2^{|S_j|-|S_{0j}\cap S_j|}  \Big\}^K  \sum_{k=1}^K 4\exp \big(- n_k \epsilon_0^2/2 \big)    \\
		&\lesssim& \Big\{ \frac{{\pi}(S_j) }{{\pi}(S_{0j}) } \Big\}^{ K} \Big\{ \nu_1^{|S_{0j}| - |S_j|} \nu_2^{|S_j|-|S_{0j}\cap S_j|}  \Big\}^K    \exp \Big\{ - (|S_{0j}| - |S_{j} \cap S_{0j}|  )  C_{\rm bm}    K  \log p  \Big\} 
		\eea
		where $N_{j, S_{kj}}$ is the set defined in the proof of Theorem 3.1 in \cite{lee2019minimax}, $\nu_1=(1+ \alpha/\gamma)^{1/2}$ and $\nu_2 = \{1- (\alpha + \nu_0/n)  / (1- 4 \sqrt{\epsilon'} - 5 \epsilon')\}^{-1/2}$.
		Note that the last inequality holds due to $  K \log p = o( \min_k n_k )$ and
		\bea
		&& \exp \Big\{ - \frac{\alpha(1-\alpha)}{4}\frac{\epsilon_0^2(1-2\epsilon_0)^2}{4}   \sum_{k=1}^K n_k \| a_{0k, S_{0j}\cap S_{j}^c} \|_2^2 \Big\} \\
		&\le& \exp \Big\{ - \frac{\alpha(1-\alpha)}{4}\frac{\epsilon_0^2(1-2\epsilon_0)^2}{4}(|S_{0j}| - |S_{j} \cap S_{0j}|  ) \min_{(j,l): a_{01,jl} \neq 0 }  \sum_{k=1}^K n_k a_{0k, jl}^2  \Big\} \\
		&\le& \exp \Big\{ - (|S_{0j}| - |S_{j} \cap S_{0j}|  )  C_{\rm bm}  K \log p  \Big\}  
		\eea		
		for some constant $C>0$ by condition (B3). 		
		Thus, \eqref{rem_o1} is bounded above by
		\bea
		&&\sum_{j=2}^p \sum_{S_{j} \nsupseteq S_{0j} } \Big\{ \frac{{\pi}(S_j) }{{\pi}(S_{0j}) } \Big\}^{K} \Big\{ \nu_1^{|S_{0j}| - |S_j|} \nu_2^{|S_j|-|S_{0j}\cap S_j|}  \Big\}^K \\
		&&\quad\quad \times  \exp \Big\{ - (|S_{0j}| - |S_{j} \cap S_{0j}|  )  C_{\rm bm}   K\log p + K(K-1)  \Big\}  \\
		&\lesssim& \exp  \big\{ -(C_{\rm bm} -c_1 -2) K \log p + K (K-1)  \big\}  \,\, =\,\, o(1) ,
 		\eea
		because $C_{\rm bm} > c_1+2$ and $K=o(\log p)$.
		This completes the proof.		
	\end{proof}

	\begin{proof}[Proof of Theorem \ref{thm:advan3}]
		Similarly to the proof of Theorem \ref{thm:advan2}, we have 
		\bea
		\sum_{j=2}^p \bbE_0 \Big\{  \tilde{\pi}_\alpha^* ( S_{j} \supsetneq S_{0j} \mid \tilde{\bfX}_n )  \Big\} 
		&\le& \sum_{j=2}^p \sum_{S_{j} \supsetneq S_{0j} }  \bbE_0 \Big\{ \prod_{k=1}^K \frac{\tilde{\pi}_\alpha^{I,*} ( S_{kj} = S_j \mid \bfX_{n_k} )}{\tilde{\pi}_\alpha^{I,*} (  S_{kj} = S_{0j} \mid \bfX_{n_k} )}  \Big\}\frac{\tilde{f}(S_j,\ldots, S_j )}{\tilde{f}(S_{0j},\ldots, S_{0j} )}  \\
		&\lesssim&  \sum_{j=2}^p \sum_{S_{j} \supsetneq S_{0j} }    \Big\{ \frac{\tilde{\pi}(S_j) }{\tilde{\pi}(S_{0j})} \Big\}^{ K} c_{\alpha, \gamma}^{K(|S_j|- |S_{0j}| ) } \exp  (K-1) \\
		&=& \sum_{j=2}^p \sum_{S_{j} \supsetneq S_{0j} }    \Big\{ \frac{{\pi}(S_j) }{{\pi}(S_{0j})} \Big\} c_{\alpha, \gamma}^{K(|S_j|- |S_{0j}| ) } \exp  (K-1)  \\
		&\lesssim&  \sum_{j=2}^p  c_{\alpha,\gamma}^{K  }  p^{-c_1  }  R_j \exp  (K-1) \\
		&\lesssim& \exp \Big\{  - (c_1-1) \log p + K \log c_{\alpha,\gamma} + K  \Big\} \,\,=\,\, o(1) ,
		\eea
		where $c_{\alpha,\gamma} = (1+ \alpha/\gamma)^{-1/2}\{2/(1-\alpha) \}^{1/2}$ and $\pi_\alpha^{I,*} (S_{kj} \mid \bfX_{n_k})\propto f(\bfX_{n_k} \mid S_{kj}) \tilde{\pi}(S_{kj})$, because $c_1>1$, $c_{2j} \le 1/(j-1)$ and $K  = o(\log p)$.
		The second and third inequalities hold by the proof of Lemma 6.1 in \cite{lee2019minimax} and
		\bea
		\frac{\tilde{f}(S_j,\ldots, S_j )}{\tilde{f}(S_{0j},\ldots, S_{0j} )}
		&\le& \exp \big\{ c_{2j} (K-1) \big| |S_j|- |S_{0j}|  \big|  \big\} \,\,\le\,\, \exp  (K-1) .
		\eea
		
		Furthermore, 
		\bean
		&& \sum_{j=2}^p \bbE_0 \Big\{  \tilde{\pi}_\alpha^* ( S_{j} \nsupseteq S_{0j} \mid \tilde{\bfX}_n )  \Big\} \nonumber \\
		&\lesssim& p K \exp \Big\{  - \frac{(\epsilon ')^2 \epsilon_0^2 }{128(1+2\epsilon_0)^2 }  \min_k n_k \Big\} \nonumber \\
		&&+ \sum_{j=2}^p \sum_{S_{j} \nsupseteq S_{0j} } \bbE_0 \Big\{ \prod_{k=1}^K  \frac{ \tilde{\pi}_\alpha^{I,*} (S_{kj}= S_j \mid \bfX_{n_k} ) }{\tilde{\pi}_\alpha^{I,*} (S_{kj}= S_{0j} \mid \bfX_{n_k}) } I( \bfX_{n_k} \in N^c_{S_{j}, \alpha, \chi^2 }   )  \Big\} \frac{\tilde{f}(S_j,\ldots, S_j )}{\tilde{f}(S_{0j},\ldots, S_{0j} )}  , \label{rem_o2}
		\eean
		where
		\bea
		&&  \bbE_0 \Big\{ \prod_{k=1}^K \frac{\tilde{\pi}_\alpha^{I,*} (S_{kj}= S_j \mid \bfX_{n_k} ) }{\tilde{\pi}_\alpha^{I,*} (S_{kj}= S_{0j} \mid \bfX_{n_k}) } I( \bfX_{n_k} \in N^c_{S_{j}, \alpha, \chi^2 }   )  \Big\}  \\
		&\lesssim& \frac{{\pi}(S_j) }{{\pi}(S_{0j}) }  \Big\{ \nu_1^{|S_{0j}| - |S_j|} \nu_2^{|S_j|-|S_{0j}\cap S_j|}  \Big\}^K    \exp \Big\{ - (|S_{0j}| - |S_{j} \cap S_{0j}|  )  C_{\rm bm}    \log p  \Big\}  ,
		\eea
		by the similar arguments used in the proof of Theorem \ref{thm:advan2} and condition (C3).
		Note that the last inequality holds due to $\log p = o( \min_k n_k )$.

		Therefore, \eqref{rem_o2} is bounded above by
		\bea
		&&\sum_{j=2}^p \sum_{S_{j} \nsupseteq S_{0j} }  \frac{{\pi}(S_j) }{{\pi}(S_{0j}) }  \Big\{ \nu_1^{|S_{0j}| - |S_j|} \nu_2^{|S_j|-|S_{0j}\cap S_j|}  \Big\}^K \\
		&&\quad\quad \times  \exp \Big\{ - (|S_{0j}| - |S_{j} \cap S_{0j}|  )  C_{\rm bm}     \log p + (K-1) \Big\}  \\
		&\lesssim& \exp  \big\{ -(C_{\rm bm} -c_1 -2)   \log p + 2K   \big\}  \,\, =\,\, o(1) ,
		\eea
		because $C_{\rm bm} > c_1+2$ and $K=o(\log p)$.
		This completes the proof.		
	\end{proof}

	\bibliographystyle{dcu}
	\bibliography{refs}
	
\end{document}